\documentclass[12pt]{article}

\usepackage{latexsym,amssymb,amsfonts,amsmath,amsthm}
\usepackage{enumitem}
\usepackage{graphicx}
\usepackage{comment}
\usepackage[normalem]{ulem}
\usepackage{subcaption}

\newtheorem{theorem}{Theorem}
\newtheorem{lemma}{Lemma}
\newtheorem{corollary}{Corollary}
\newtheorem{proposition}{Proposition}
\newtheorem{remark}{Remark}
\newtheorem{definition}{Definition}

\newtheorem{assumption}{Assumption}
\newtheorem{fact}{Fact}

\numberwithin{theorem}{section}
\numberwithin{lemma}{section}
\numberwithin{corollary}{section}
\numberwithin{proposition}{section}
\numberwithin{remark}{section}
\numberwithin{definition}{section}
\numberwithin{example}{section}
\numberwithin{fact}{section}

\newcommand{\bs}[1]{\boldsymbol{#1}}

\newcommand{\supp}{{\rm supp}}
\newcommand{\tr}{\rm tr}

\newcommand{\bra}{\langle}
\newcommand{\ket}{\rangle}
\newcommand{\brara}{\langle\langle}
\newcommand{\kett}{\rangle\rangle}
\newcommand{\mL}{\mathcal{L}}

\usepackage{color}

\title{Interpolating walk between discrete-time quantum walk and its intrinsic random walk on graph}
\author{Yusuke Higuchi$^{1}$, Etsuo Segawa$^2$, Honoka Shiratori$^{3}$, Saori Yoshino$^{2}$\\
$^1${\small Department of Mathematics,
Gakushuin University,} \\
{\small Toshima, Tokyo, 171-8588, Japan}\\
$^2${\small Graduate School of Environment and Information Sciences, Yokohama National University,}\\ 
{\small Hodogaya, Yokohama, 240-8501, Japan}\\
$^3${\small Department of Information Physics and Computing,  
The University of Tokyo}\\
{\small Bunkyo, Tokyo, 113-8656, Japan 
}}

\date{}

\begin{document}

\maketitle
\noindent {\bf Abstract}
We consider an interpolation with the parameter $p\in [0,1]$ between the discrete-time quantum ($p=0$) and its intrinsic random ($p=1$) walks on a finite connected graph. Its time evolution is defined by the convex combination of the Kraus (CPTP) maps of the quantum and random walks. 
The Kraus map of the random walk is represented by taking a kind of projection of that of the quantum walk. The time evolution of the interpolation is interpreted as the combination of the following two dynamics of $2$ walkers on the same graph correlated with each other: $2$ walkers move independently of each other (no-correlation) with probability $1-p$, while $2$ walkers always move into the same position 
(the maximal correlation) with probability $p$ at each time step. In this paper, we generalize the random walk to a new walk, namely, the correlated walk, which preserves the intrinsic property and relaxes the maximal correlation of the random walk. This correlation is determined by a partition of the arc sets in the underlying graph. We show that the eigenvalues of the interpolating walk between the correlated and quantum walks live in $\{ z\in \mathbb{C}\;|\;1-p\leq |z|\leq 1 \}$ and that the absorption state coincides with that of the correlated walk, which is characterized by the underlying graph structures. \\

\noindent{\it Key words and phrases.} 
Convex combination of Quantum and Random walks, 
Arc coloring, Absorption state
\\


\section{Introduction}
Discrete-time quantum walks are quantum dynamical analogs of random walks~\cite{Meyer, Am1} that can be applied to quantum computing~\cite{Am2,Portugal}. 
The time evolution of the unitarity ensures full incoherence and shows counterintuitive behavior from the viewpoint of random walks~\cite{Konno05, HKKSS17}. 
The study of quantum walks with incoherence is one of the interesting topics that combine quantum and random walks~\cite{SPYGJ, YoshinoEtAl}. 
Several kinds of such models have been constructed, for example, 
open quantum random walks~\cite{AttalEtAl} and  continuous-time stochastic quantum random walks on graphs~\cite{WRA}. 
These models are motivated not only by  such mathematical considerations but also by applications to quantum simulation of photon synthesis~\cite{rmkla} and quantum neural networks~\cite{SSIP}. 
A difference method for quantum stochastic walks on graphs~\cite{WRA} is presented by \cite{sgtw}. The resulting discrete time evolution can be described as a convex combination of quantum and random walks. 

In the following, let us explain an {\it intrinsic property of random walks} in discrete-time quantum walks, which is the key idea of this paper.  
Let $\Omega$ and $U$ be a finite countable set and a unitary matrix on $\mathbb{C}^\Omega$, respectively. 
The {\it state} considered here is represented by the following matrix labeled by $\Omega$:  
\[ S(\mathbb{C}^\Omega):=\{ \rho\in  \mathbb{C}^{\Omega\times \Omega}\;|\;\rho\geq 0,\;\tr(\rho)=1 \}. \]
The discrete-time quantum walk on $\Omega$ is realized by the transition of the state with the iteration of $\mathcal{L}_{QW}: S(\mathbb{C}^\Omega)\to S(\mathbb{C}^\Omega)$ such that 
\begin{align*}
\rho_0 &\in S(\mathbb{C}^\Omega),\\
\rho_{t}&=U\rho_{t-1}U^*:=\mathcal{L}_{QW} \rho_{t-1}\;\;(t=1,2,\dots). 
\end{align*}
Here, we should note that the finding probability at time $t$ and $\omega\in \Omega$ is described by $\rho_t(\omega,\omega)$. 
The intrinsic property can be seen by dividing the state $\rho_t$, which is represented by $|\Omega|\times |\Omega|$ matrix, into the diagonal and off-diagonal parts such that 
\begin{align*} 
\rho_t&= \sum_{\omega} \Pi_\omega U\rho_{t-1} U^*\Pi_{\omega}
+\sum_{\omega\neq \omega'} \Pi_\omega U\rho_{t-1} U^*\Pi_{\omega'}.
\end{align*}
Let us pick up the diagonal operation of the first term in the RHS, $\mathcal{L}_{diag}$, and consider the single iteration of $\mathcal{L}_{\text{diag}}$. The definition of $\mathcal{L}_{\text{diag}}$ implies that  
\[ (\mathcal{L}_{\mathrm{diag}}\nu_{t})(\omega,\omega')=\delta_{\omega,\omega'}\sum_{\omega''} |(U)_{\omega,\omega''}|^2\nu_{t-1}(\omega'',\omega'')\;\;(t=1,2,\dots) \]
with the initial state $\nu_0\in S(\ell^2(\Omega))$. This expression highlights the following: By the unitarity of $U$, the Hadamard product $U\circ \bar{U}$, $(U\circ \bar{U})_{\omega, \omega'}=|(U)_{\omega, \omega'}|^2$ ($\omega,\omega'\in \Omega$), describes a (doubly) stochastic matrix on $\Omega$, that is, $\sum_{\omega}|(U)_{\omega,\omega'}|^2(=\sum_{\omega'}|(U)_{\omega,\omega'}|^2)=1$. Therefore, the Kraus 
(CPTP) map $\mathcal{L}_{diag}=:\mathcal{L}^{RW}$ presents the random walk on $\Omega$, where the transition probability from $\omega'$ to $\omega$ at each time step is given by $|(U)_{\omega,\omega'}|^2$. In other words, the quantum walk $\mathcal{L}^{QW}$ contains a doubly stochastic random walk; that is, \[\mathcal{L}_{QW}=\mathcal{L}_{RW}+\mathcal{L}_{\text{off-diag}}.\] 
This is the intrinsic property of random walks in the discrete-time quantum walk. 
Then, let us control the coherence $\mL_{\text{off-diag}}$ of the quantum walk by setting the thinning parameter $q\in[0,1]$ with
\[ \mathcal{L}_q:=\mathcal{L}_{RW}+q \mathcal{L}_{\text{off-diag}}. \]
Note that $\mathcal{L}_q$ is still a CPTP map since it can be written as a convex combination of the CPTP maps $\mathcal{L}_{QW}$ and $\mathcal{L}_{RW}$, that is, $\mathcal{L}_q=(1-q)\mathcal{L}_{RW}+q\mathcal{L}_{QW}$.
It is natural to ask how the finding probability $\rho_t(\omega,\omega)$ depends on the parameter $q\in[0,1]$. 
The stability of ergodicity under such randomization is discussed in \cite{BCGPY}.  
We focus on the dynamics on a finite and symmetric digraph $G=(X,A)$ driven by $\mathcal{L}_q$ with $\Omega=A$ and   discuss how graph structures reflect the stationary state by discussing the spectral properties of $\mL_q$. 

In this paper, we generalize $\mL_q$ to $\mL_{p,\pi}$, which is an interpolation between $\mL_\pi$ and $\mL_{QW}$ by a parameter $p\in[0,1]$. 
Here, $\mL_\pi$ is a correlated walk, which is defined in terms of the partition $\pi$ of $A$. 
The details are provided in Section~2. 
if we choose the trivial partition $\pi_o: A=\{A\}$, then $\mL_{\pi_o}=\mL_{QW}$; 
if we divide $A$ into the smallest parts $\pi_{max}: A=\sqcup_{a\in A}\{a\}$,  then $\mL_{\pi_{max}}=\mL_{RW}$. 
In fact, for a partition $\pi:A=\sqcup_{h}A_h$, $\mL_\pi$ is represented as 
\[ \mL_\pi\rho=\sum_{h}\Pi_h U\rho U^*\Pi_h. \]
Let us consider $\rho_t$ which is the density matrix of a correlated walk $\mL_{\pi}$ at time $t$. 
The value $\rho_t(a,b)$ can be interpreted as the complex-valued amplitude at time $t$ when $2$ walkers are located in arcs $a$ and $b$, respectively. 
This time evolution implies that if the first walker chooses one element of $A_h$, then the second walker must also choose one element of the same $A_h$. 
Then, if we choose the partition $\pi=\pi_o$, the $2$ walkers move independently without any correlation to each other. If we choose the partition $\pi=\pi_{max}$, the $2$ walkers must move to the same position with full correlation to each other. 
In this sense, a correlated walk $\mL_\pi$ is a kind of intermediate walk between $\mL_{QW}$ and $\mL_{RW}$. 

Thus, we focus on the operator for a fixed partition $\pi$ on $A$, whose form is 
\[ \mL_{p,\pi}=p\mL_{\pi}+(1-p)\mL_{QW}, \]
where $p\in[0,1]$, and discuss the spectral properties of $\mL_{p,\pi}$ and the asymptotic behavior of $\rho_t$ when $t\to\infty$ in Theorem~\ref{thm:abs}. 
Moreover, for $\pi=\pi_{max}$, and for $\pi=\pi_c$ which is induced by an arc coloring and different from $\pi_{o}$ and $\pi_{max}$ (see Section~2.3), we provide the concrete forms of $\rho_t$ when $t\to\infty$ in Theorems~\ref{thm:max} and \ref{prop:abs}, respectively. 

The remainder of this paper is organized as follows. 
In Section~2, we provide the general settings and notations for stating our results. 
First, we give the definition of partitions and a non-trivial partition $\pi_c$ for $d$-arc colorable graphs. Subsequently, we introduce the correlated walk $\mL_{\pi}$ and provide the definition of the interpolating walk $\mL_{p,\pi}$.
In Section~3, the spectral information of the interpolating walk $\mL_{p,\pi}$ is discussed. We show that the eigenvalues are contained between the circles whose radii are $1-p$ and $1$ in the complex plane.  
In Section~4, we characterize the absorption space for $\mL_{p,\pi}$, which plays an important role in discussing the asymptotic behavior of $\rho_t$. 
Moreover, we give detailed expressions in the cases where $\pi=\pi_{max}$ and $\pi=\pi_c$. 
Therefore, we can characterize the stationary state by the equivalence class of dynamical graphs for arc colorable graphs. Finally, we provide the summary and discussion through numerical simulations in Section~5. 
\section{Setting}
\subsection{Setting of graph}
Let $G=(X,A)$ be a finite, connected symmetric digraph. 
For any arc $a\in A$, the {\it inverse arc} is denoted by $\bar{a}$. Note that  $a\in A$ implies  $\bar{a}\in A$ since $G$ is a symmetric graph. 
The {\it origin} and {\it terminal} vertices of $a\in A$ are denoted by $o(a),t(a)\in X$, respectively. 
Remark that $t(\bar{a})=o(a)$ holds for any $a\in A$. The {\it support} of $a\in A$ is denoted by $|a|$ with $|a|=|\bar{a}|$. Set the undirected edge set $E$ by $E=\{|a| \;:\;a\in A\}$. The degree of $x\in X$ is defined by $\mathrm{deg}(x)=|\{a\in A\;:\;t(a)=x\}|$. 
If $G$ satisfies $\mathrm{deg}(x)=d$ for any $x\in X$, $G$ is said to be $d$-regular. 
\subsection{Partition of arcs}\label{sect:pa} 
Let $\pi$ be a partition of $A$, such that \[\pi: A=\bigsqcup_{h} A_h,\] where $A_h \cup A_{h'}=\emptyset$ for $h\neq h'$. 
We set the trivial partition by \[\pi_{0}: A=\{A\}\] while the maximal partition by 
\[\pi_{max}:A=\bigsqcup_{a\in A}\{a\}.\]  
To set an intermediate partition between $\pi_0$ and $\pi_{max}$, let us introduce the notion of $d$-arc colorable. 
See also Fig.~\ref{fig:arccolorling}. 
\begin{definition}[$d$-arc colorable]\label{def:arcc}
We say that $G$ is {\it $d$-arc colorable} if $G$ is a $d$-regular graph and there exist a map from arcs to colors $\gamma:A\to \{1,2,\dots,d\}$ and an involution map $\phi:\{1,2,\dots,d\}\to \{1,2,\dots,d\}$ with $\phi^2=id$, such that 
\begin{enumerate}
\item 
The colors of arcs with the same terminal vertex must be different, that is, \\
for any $a,b\in A$, if $t(a)= t(b)$, then $\gamma(a)\neq \gamma(b)$;
\item 
The colors of inverse arcs of the same colored arcs (e.g., color $i$) is also the same (e.g., color $j$) regardless of the arcs, that is, 
\[
\gamma(\bar{a})=\phi(\gamma(a))
\]
for any $a\in A$. 
\end{enumerate}
\end{definition}
For a $d$-arc colorable graph $G$, we fix the pair the map $\gamma$ in (1) and the involution $\phi$ in (2); we call the pair of $(\gamma,\phi)$ a {\it coloring}. 
Here, we remark that $(\gamma,\mathrm{id})$ gives a usual edge coloring in graph theory, if exists.  
Let $H=\{1,2,\dots,d\}$ be the set of colors. 
For every arc $a\in A$, the pair $(t(a),\gamma(a))\in X\times H$ is uniquely determined. Conversely, for every $(x,h)\in X\times H$, the corresponding arc $a\in A$ is uniquely determined by $x=t(a)$ and $h=\gamma(a)$. Thus, we have $A\cong X\times H$.
Under this identification, we set the partition $\pi_c$ of $A$ by 
\[ \pi_c:A=\bigsqcup_{h\in H}\gamma^{-1}(h)=\bigsqcup_{h\in H}\left(X\times\{h\}\right). \]

Let us say that $G=(X,A)$ has an underlying Cayley graph structure if $X=\bra H \ket$ and $a\in A \Leftrightarrow$ ``there exists an $h\in H$ such that $t(a)=h o(a)$", where $H$ is a generator of a group. If $G$ can be treated as a Cayley graph, then a coloring $(\gamma,\phi)$ can be naturally set by $\gamma(a)=t(a)o(a)^{-1}\in H$ and $\phi(h)=h^{-1}$. See the middle coloring in Figure~\ref{fig:arccolorling}. 
As a typical example, let us see that the cycle $C_N$ ($N\geq 3$) is $2$-arc colorable. 
The vertex set of $C_N$ can be represented by $X=\mathbb{Z}_N=\langle \pm 1 \rangle$ as an abelian group,  and the arc set can be represented by
\[A=\{(x;L)\;:\;x\in X\}\sqcup \{(x;R)\;:\;x\in X\},\]
where $(x;R)$ is the arc whose terminal and origin vertices are $x$ and $x-1$, 
$(x;L)$ is the arc whose terminal and origin vertices are $x$ and $x+1$. 
The above labeling of arcs naturally gives the following coloring $(\gamma,\phi)$ with 
\[ \gamma(a)=\begin{cases} 1 & \text{: $t(a)-o(a)=+1(\in H)$,} \\ -1 & \text{: $t(a)-o(a)=-1(\in H)$,} \end{cases} \]
for any $a\in A$ and 
\[ \phi(1)=-1,\;\phi(-1)=1.  \]

Finally, let us state the characterization of the arc colorable as follows. 
\begin{proposition}
Let $G$ be a $d$-regular graph. 
For even $d$, $G$ is $d$-arc colorable. 
For odd $d$, $G$ is $d$-arc colorable if and only if there is a perfect matching in $G$.
\end{proposition}
\begin{proof}

A $2$-factor is a spanning subgraph of 
$G$ whose connected components collection of cycles.
The famous $2$-factor theorem by J. Petersen~\cite{Petersen} is that the edges of every $2d'$ regular graph can be partitioned into $d'$ edge-disjoint $2$-factors. For each cycle $C$ in a $2$-factor, we can assign one color to each arc in a directed cycle $\overrightarrow{C}$ and another color in the reverse directed cycle $\overleftarrow{C}$. Then each $2$-factor gives $2$ colors for an arc coloring of $G$ and the conclusion holds.

Let us consider the case where $d$ is odd. We first assume that there is a perfect matching in $G$. All the arcs in the matching edges are colored by ``$1$". The subgraph deleted all the arcs colored by ``$1$" becomes an even regular graph. Then, applying $2$-factor theorem to this subgraph again, we can make a coloring using the other colors $\{2,\dots,d\}$. This implies that there is a coloring $\gamma$ with $\phi(1)=1$.  
Now let us assume $G$ is $(2\ell+1)$-arc colorable and $(\gamma,\phi)$ is one of the colorings. For a vertex $x\in X$, if there exists an arc with $t(a)=x$ such that $\gamma(a)\neq \gamma(\bar{a})$, then there exists an arc $t(b)=x$ with $\gamma(b)=\gamma(\bar{a})$ by the definition of $(\gamma,\phi)$. Note that $\gamma(\bar{b})=\gamma(a)$. 
Thus, there must exist a color $j$ and an arc $a$ with $t(a)=x$ such that $\gamma(a)=\gamma(\bar{a})=j$ because $\mathrm{deg}(x)=2\ell+1$ is odd.  
By setting $M=\{ |a| \;|\; \gamma(a)=\gamma(\bar{a})=j \}$ in $G$, $M$ becomes a $1$-factor in $G$, that is, a perfect matching. 
\end{proof}
\begin{figure}
\centering
\includegraphics[width=150mm]{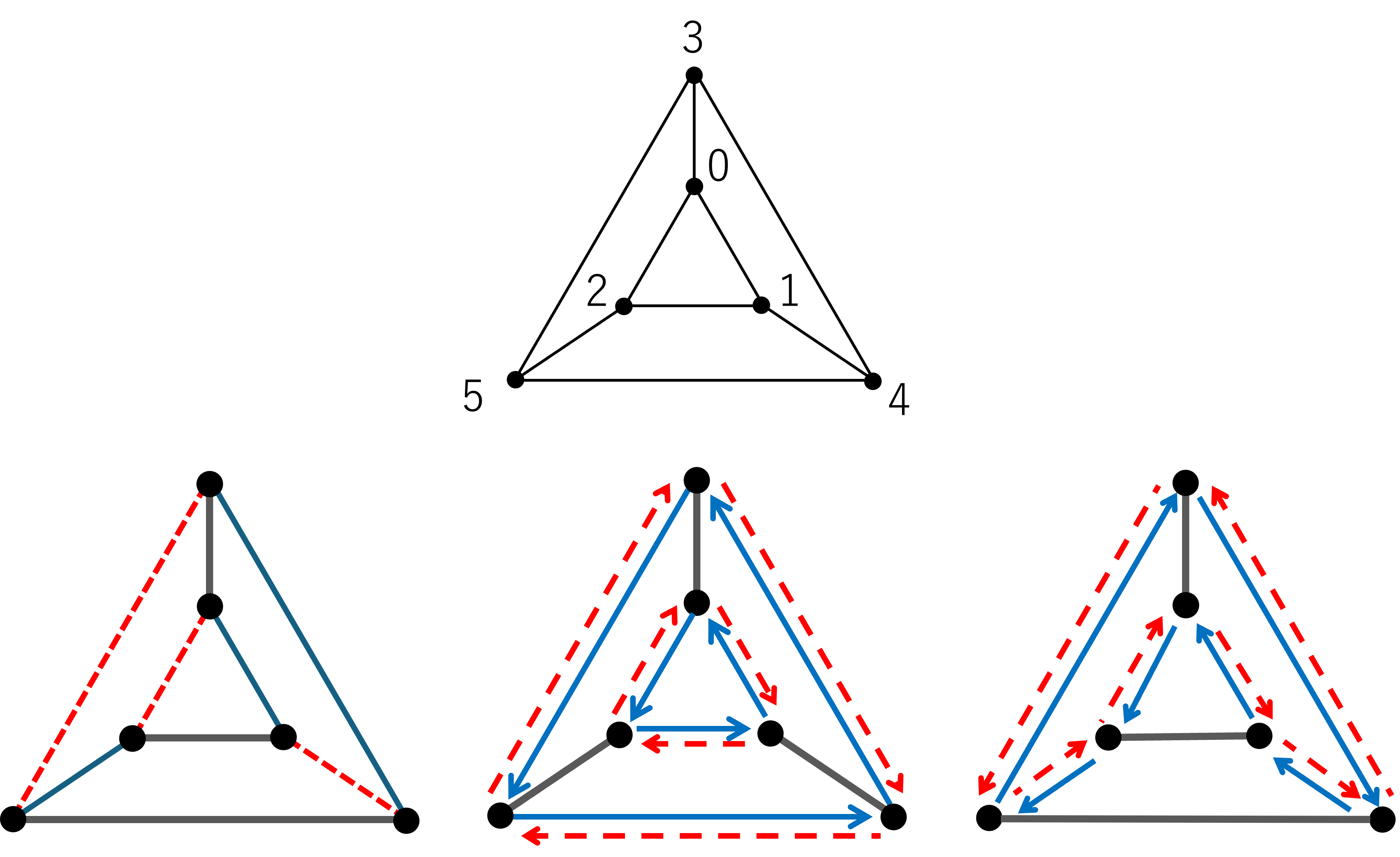}
    \caption{Arc colorings: The top graph is the original graph. There are essentially the above three coloring up to the isomorphism in this graph. The coloring in the left figure is represented by $(\gamma,id)$ which is isomorphic to the edge coloring; the non-directed edges depict the double headed arcs. The middle coloring is represented by $(\gamma', (\text{red}, \text{blue}))$; each color corresponds to the generators of the dihedral group $D_3=\bra z, \tau\;| z^3=\tau^2=(z\tau)^2=1,\; \ket$, that is, $z=\text{red}$, $z^{-1}=\text{blue}$ and $\tau=\text{black}$ in the sense of graph isomorphism. The coloring in the right figure is represented by$(\gamma'',(\text{red}, \text{blue}))$, which is not isomorphic to either of the above two colorings.}
    \label{fig:arccolorling}
\end{figure}

\subsection{Setting of correlated walk}
For a countable set $\Omega$, $\mathbb{C}^\Omega$ denotes the vector space whose standard base is labeled by each element of $\Omega$. 
Here, the time evolution operator of the discrete-time quantum walk on $G=(X,A)$ is denoted by a unitary operator $U$ on $\ell^2(A)=\{\psi\in \mathbb{C}^A\;;\;||\psi||<\infty\}$ with the standard inner product. This operator $U$ is determined by the sequence of $\mathrm{deg}(x)$-dimensional unitary operators $\{C_x\}_{x\in X}$ such that 
\begin{equation}\label{eq:QW}
U=SC, 
\end{equation}
where 
\[ S\cong \bigoplus_{e\in E} \begin{bmatrix} 0 & 1 \\ 1 & 0\end{bmatrix} \]
under the decomposition of $A=\bigsqcup_{e\in E}\{a\;:\; |a|=e\}$
and 
\[ C\cong \bigoplus_{x\in X}C_x \]
under the decomposition of $A=\bigsqcup_{x\in X}\{a\;:\;t(a)=x\}$.
The $\mathrm{deg}(x)$-dimensional unitary matrix $C_x$ is called the local coin matrix assigned at  $x\in X$. 
Set $\{a\in A\;:\;o(a)\}=\{a_1,\dots,a_r\}$. Then 
the definition of $U$ implies 
\[ \begin{bmatrix} 
(U\psi)(a_1) \\ \vdots \\ (U\psi)(a_r)
\end{bmatrix}=C_x\begin{bmatrix}\psi(\bar{a}_1) \\ \vdots \\ \psi(\bar{a}_r)\end{bmatrix}, \]
which gives the local scattering at the vertex $x$. 

For a given finite connected symmetric digraph $G=(X,A)$, the total state space considered in this paper is given not directly by $\ell^2(A)$ but by the set of linear operators on $\ell^2(A)$ satisfying the following properties: (i) the semi-positivity; $\langle \psi,\rho\psi \rangle\geq 0$ for any $\psi\in \ell^2(A)$, and (ii) the normalized trace; $\mathrm{tr}(\rho)=1$, that is,  
\[ \mathcal{H}_G=\mathcal{S}(\ell^2(A))=\{\rho\in \mathbb{C}^{A\times A} \;;\;\rho\geq 0,\;\mathrm{tr}(\rho)=1\}.  \]
We regard $\rho(a,b)$ as the state that walkers $1$ and $2$ exist at the arcs $a$ and $b$, respectively.  
To make a correlation between walkers $1$ and $2$ through a time evolution, we use a  partition of $A$, $\pi'$ so that two walkers move in the same subset of $A$ as follows.  
The time evolution operator of the correlated walk of $\pi'$ is denoted by  \begin{equation}\label{eq:proh} \mathcal{L}_{\pi'}\rho=\sum_{h}\Pi_h U\rho U^*\Pi_h, 
\end{equation}
for any $\rho\in \mathcal{S}(\ell^2(A)$, 
where $\Pi_h$ is the projection operator onto 
$\{ \delta_a \;:\; a\in A_h  \}\subset \ell^2(A)$ under the partition $\pi': A=\sqcup_{h}A_h$, that is, 
\[ (\Pi_h\psi)(a)=\begin{cases} \psi(a) & \text{: $a\in A_h$, }\\ 0 & \text{: $a\notin A_h$} \end{cases} \]
for any $a\in A$ and $\psi\in \mathbb{C}^A$. 
Let us confirm that $\mathcal{L}_{\pi'}$ preserves the trace and positivity, that is, ``CPTP" follows:
\begin{align*}
\mathrm{tr}(\mathcal{L}_{\pi'}\rho) &= \sum_{h}\mathrm{tr}(\Pi_hU\rho U^*\Pi_h)
= \sum_{h}\mathrm{tr}(\Pi_hU\rho U^*) 
= \mathrm{tr}\left(\sum_{h}\Pi_hU\rho U^*\right) \\
&= \mathrm{tr}(U\rho U^*)= \mathrm{tr}(\rho U^*U) \\
&= \mathrm{tr}(\rho),  
\end{align*}
and, for any $\rho\geq 0$, 
\begin{align*}
\langle \psi,\mathcal{L}_{\pi'}\rho\psi \rangle &= \left\langle \psi,\sum_h\Pi_h U\rho U^*\Pi_h\psi\right\rangle \\
&= \sum_h \langle U^*\Pi_h\psi,\rho U^*\Pi_h\psi \rangle \\
&\geq 0
\end{align*}
Then the convex combination of CPTP maps 
\[ \mathcal{L}_{p,\pi}=p\mathcal{L}_{\pi}+(1-p)\mathcal{L}_{\pi_0} \]
is also CPTP map for any $p\in [0,1]$ and partition $\pi$. 

In summary, the time evolution of the interpolating walk is defined as follows. 
\begin{definition}[Interpolating walk between quantm walk and correlated walk]
For finite connected symmetric graph $G=(X,A)$, the partition is set as $\pi: A=\sqcup_{h}A_h$. 
\begin{enumerate}
\item The total state space: $\mathcal{H}_G=\mathcal{S}(\ell^2(A))=\{\rho\in\mathbb{C}^{A\times A}\;|\;\rho\geq 0,\;\mathrm{tr}(\rho)=1\}$. 
\item The time evolution: for each $t=1,2,\dots$, $\rho_t\in \mathcal{H}_G$ is defined by $\rho_{t}=\mathcal{L}_{p,\pi}\rho_{t-1}$ with the parameter $p\in[0,1]$ and the initial state $\rho_0\in \mathcal{H}_G$ with
$\mathrm{tr}(\rho_0)=1.$
 
Here the time evolution operator $\mathcal{L}_{p,\pi}$ is defined as follows. 
\[ \mathcal{L}_{p,\pi}=p\mathcal{L}_{\pi}+(1-p)\mathcal{L}_{\pi_0}, \]
where for a partition $\pi'\in \{\pi,\pi_0\}$, 
\[ \mathcal{L}_{\pi'}\rho= \sum_{h}\Pi_h U\rho U^*\Pi_h. \]
\item The finding probability: 
the finding probability at time $t\in \{0,1,2,\dots\}$ and arc $a\in A$ $\mu_t:A\to[0,1]$ is defined by 
\[ \mu_t(a)=\rho_t(a,a).  \]
\end{enumerate}
\end{definition}
\begin{remark}
If $p=0$, and the initial state $\rho_0$ is chosen as a pure state, that is, $\rho_0(a,b)=\psi_0(a)\bar{\psi}_0(b)$ with $||\psi_0||_{\mathbb{C}^A}^2=1$ for any $a,b\in A$, then the finding probability at each time $t=0,1,2,\dots$ of the discrete-time quantum walk on $G$ driven by $U$ with the initial state $\psi_0$, $|(U^t\psi_0)(a)|^2$, is reproduced by $\mu_t(a)=|(U^t\psi_0)(a)|^2$ for any $a\in A$. 
\end{remark}
Then, from now on, let us describe $\mathcal{L}_{\pi_0}$ by $\mathcal{L}_{QW}$. 
\begin{remark}\label{rem:CRW}
Let us see that the case for $\pi=\pi_{max}$ with $p=1$ reproduces the intrinsic random walk on $G$, which is doubly stochastic. 
Note that if $a\neq b$, then $\rho_t(a,b)=0$ for $t\geq 1$, by the definition of $\mathcal{L}_{\pi_{max}}$.
Then we have 
\begin{align*}
\rho_{t+1}(a,b)&=(\mathcal{L}_{\pi_{max}}\rho_{t})(a,b) \\
&=\delta _{a,b}\sum_{c,c'}(U)_{a,c}\rho_{t}(c,c')(U^*)_{c',a} \\
&= \delta _{a,b}\sum_{c}|(U)_{a,c}|^2\rho_{t}(c,c).
\end{align*}
By the unitarity of $U$, $\sum_{c}|(U)_{a,c}|^2=1$, then $\rho_t(a,a)=\nu_t(a)$ follows the intrinsic random walk on $G$ driven by the (probability) transition matrix $T=U\circ \bar{U}$, where ``$\circ$" means the Hadamard product, and $\bar{U}$ means the complex conjugate matrix of $U$. 
If we regard this random walk as the dynamics on vertices, then the transition probability at each time step depends on the position of a random walker at the previous step. 
More precisely, let us consider, for each vertex $x\in X$, 
$$
\tilde{\nu_{t}} (x) = \sum_{t(a)=x}\rho_{t}(a,a).
$$
Thus $\tilde{\nu_{t}}$ presents a dynamics on vertices 
induced by $T$ stated above.
For the standard random walk on vertices, 
the probability of the choice of neighborhood 
in the next step $t+1$ is determined by 
the position $x\in X $ at time $t$. 
In the contrast, for this dynamics, 
we find the choice of neighborhood depends on 
the arcs $a$'s with $t(a)=x$ on the preceeding step. 
Such a random walk is a special case of the persistent, Goldstein-Kac or correlated random walks~i.g.,\cite{Goldstein,Kac,RH,W}; in particular, if the underlying graph $G$ is the $1$-dimensional lattice, then the limit theorem for $\mu_t$ is essentially obtained by \cite{KonnoRW}.  
\end{remark}
\begin{remark}
The time evolution driven by $\mathcal{L}_{p,\pi}$ gives an average of a  randomly chosen time evolution operator from $\{\mathcal{L}_\pi, \mathcal{L}_{QW}\}$. More precisely,  
set the i.i.d. Bernoulli sequence $\omega=(\omega(1),\omega(2),\dots,\omega(n))$ such that 
$\omega(t)\in\{\pi,QW\}\;(t=1,2,\dots,n)$ with 
\[P[\omega(j)=\pi]=p,\;P[\omega(j)=\pi_0]=1-p.  
\]  
We define 
\begin{equation*}
    \rho_n^\omega:=\mathcal{L}_{\omega(n-1)}\rho_{n-1}^\omega
\end{equation*} 
with $\rho_0^\omega=\rho_0$. 
By taking the average over $\omega\in\{\pi,QW \}^n$, 
we have  
\begin{equation*}
E[\rho_n^\omega]
=E[\mathcal{L}^{\omega(n-1)}\rho_{n-1}^\omega]
=((1-p)\mathcal{L}_{QW} + p\mathcal{L}_{\pi})\;E[\rho_{n-1}^\omega],
\end{equation*}
which coincides with the time iteration of $\mathcal{L}_{p,\pi}$. 
\end{remark}

\section{Spectral properties of $\mathcal{L}_{p,\pi}$}
For any $\rho_1, \rho_2\in \mathbb{C}^{A\times A}$, we introduce the following inner product of matrices, which is equivalent to the standard inner product of the vectors on $\mathbb{C}^A\otimes \mathbb{C}^A$ and also coincides with the Hilbert-Schmidt inner product:
\[ \brara \rho_1,\rho_2\kett=\sum_{a,b}\bar{\rho}_1(a,b)\rho_2(a,b)=\mathrm{tr}(\rho_1^*\rho_2). \]
Here, $\rho^*$ is the complex conjugate transpose of $\rho$ such that $\rho^*(a,b)=\overline{\rho(b,a)}$. 
In particular, we denote $\brara \rho_1,\rho_1 \kett=:\brara \rho_1 \kett$. 
The properties of $\brara \cdot,\cdot \kett$ are listed as follows: 
\begin{enumerate}
\item Writing the matrices $\rho_1$ and $\rho_2$ by column vectors 
\[\rho_1=\begin{bmatrix}\;\bs{\rho}_1^{(1)}&|\;\bs{\rho}_1^{(2)}&|\;\cdots&|\;\bs{\rho}_1^{(|A])}\;\end{bmatrix},\;\;  \rho_2=\begin{bmatrix}\;\bs{\rho}_2^{(1)}&|\;\bs{\rho}_2^{(2)}&|\;\cdots&|\;\bs{\rho}_2^{(|A|)}\;\end{bmatrix},\] we have 
\[ \brara \rho_1,\rho_2 \kett=\sum_{j=1}^{|A|}\bra\bs{\rho}_1^{(j)},\bs{\rho}_2^{(j)}\ket. \]
Here $\bra \cdot,\cdot\ket$ is the standard inner product on $\mathbb{C}^{A}$. 
\item For any $M\in \mathbb{C}^{A\times A}$, 
\begin{align*}
\brara \rho_1,M\rho_2 \kett &=\brara M^*\rho_1,\rho_2 \kett \\
\brara \rho_1,\rho_2M \kett &=\brara \rho_1M^*,\rho_2 \kett 
\end{align*}
\item For any $\rho, \rho_1,\rho_2\in \mathbb{C}^A$, 
\[\brara \rho,\rho_1+\rho_2 \kett=\brara \rho,\rho_1 \kett+\brara \rho,\rho_2 \kett. \]
\item For any $\rho_1,\rho_2\in \mathbb{C}^A$, 
\[ \brara \rho_1^*,\rho_2^* \kett=\brara \rho_2,\rho_1 \kett. \]
\end{enumerate}

Using the above notation, we prepare properties $\mathcal{L}_{p,\pi}$ in the following lemmas.
\begin{lemma}\label{lem:1}
For any $\rho\in \mathbb{C}^{A\times A}$, we have 
\[ \brara \;(\mL_{QW}-\mL_{\pi})\rho,\mL_\pi \rho\;\kett=0 \]
\end{lemma}
\begin{proof}
By using the properties of the inner product $\brara \cdot,\cdot\kett$, we obtain
\begin{align*}
\brara\;(\mL_{QW}-\mL_{\pi})\rho, \mL_{\pi}\rho \kett &= 
\brara \sum_{(h,h')\in H\times H,\;h\neq h'} \Pi_hU\rho U^*\Pi_{h'},\sum_{h''\in H} \Pi_{h''}U\rho U^*\Pi_{h''}\kett \\
&= \sum_{h\neq h'}\sum_{h''}\brara  \Pi_hU\rho U^*\Pi_{h'},\;\Pi_{h''}U\rho U^*\Pi_{h''} \kett \\
&= 0.
\end{align*}
\end{proof}
\begin{lemma}\label{lem:2}
For any $\rho\in \mathbb{C}^{A\times A}$, 
\[\brara \mL_{QW} \rho \kett=\brara\; (\mL_{QW}-\mL_{\pi})\rho \;\kett+\brara\; \mL_\pi\rho\; \kett=\brara \rho \kett. \]
\end{lemma}
\begin{proof}
The first equality immediately follows from Lemma~\ref{lem:1}. 
The last equality derives from the unitarity of $U$. 
\end{proof}

Now we give spectral properties of $\mathcal{L}_{p,\pi}$ in the following two propositions.
It is well known that the operator norm is greater than or equal to the spectral radius in general. 
Since not only the regularity, but even the diagonalizability are not ensured in $\mathcal{L}_{p,\pi}$. However, the first theorem shows the equality. 
\begin{theorem}[Operator norm and Spectral radius]\label{prop:1} 
Assume $p\neq 0$. 
\noindent
\begin{enumerate}
\item For any $\rho\in \mathbb{C}^{A\times A}$,  
\[ \brara \mathcal{L}_{p,\pi}\rho \kett \leq \brara \rho \kett. \]
\item The eigenvalues of $\mathcal{L}_{p,\pi}$ are included in the following domain.  
\[ \mathrm{spec}(\mathcal{L}_{p,\pi})\subseteq \{z\in \mathbb{C}\;:\;1-p\leq |z|\leq 1\}. \]
In particular, 
\begin{itemize}
\item if $|\lambda|=1$, then 
\[ 
\ker(\lambda-\mathcal{L}_{p,\pi})=\ker(\lambda-\mathcal{L}_{\pi}) \]
and $1\in \mathrm{spec}(\mathcal{L}_{p,\pi})$ with $I_A\in \ker(1-\mL_{p,\pi})$,  \\
\item if $|\lambda|=q$, then 
\[ \ker(\lambda-\mathcal{L}_{p,\pi})=
\ker\left(\frac{\lambda}{|\lambda|}-(\mathcal{L}_{QW}-\mL_\pi)\right)
\]
\end{itemize}
\end{enumerate}
\end{theorem}
\begin{proof}
\noindent 
\begin{itemize}
\item First part: 
Set $q=1-p$. Then for any $\rho\in \mathcal{H}$, we have   
\begin{align}
\brara \mL_{p,\pi}\rho \kett
&= \brara \mL_{\pi}\rho+q(\mL_{QW}-\mL_{\pi})\rho\kett \notag \\
&= \brara \mL_{\pi}\rho \kett+q^2\brara\; (\mL_{QW}-\mL_{\pi})\rho\; \kett. \label{eq:3}
\end{align}
Here we used Lemma~\ref{lem:1} in the last equality. 
By using Lemma~\ref{lem:2} and (\ref{eq:3}), 
we have the following other expressions for $\brara \mL_{p,\pi}\rho \kett$, which will be convenient in our discussion later: 
\begin{align}
\brara \mL_{p,\pi}\rho \kett 
&= \brara \mL_{QW}\rho \kett-(1-q^2)\brara \;(\mL_{QW}-\mL_{\pi})\rho\; \kett \label{eq:upper}\\
&=q^2\brara \mL_{QW}\rho \kett+(1-q^2)\brara \mL_{\pi}\rho \kett. \label{eq:lower}
\end{align}
By (\ref{eq:upper}) and (\ref{eq:lower}), for any $\rho\in \mathcal{H}$, we have 
\begin{align}
\brara \mL_{p,\pi}\rho\kett &\leq \brara \mL_{QW} \rho\kett=\brara \rho \kett, \label{eq:upper1}\\
\brara \mL_{p,\pi}\rho\kett &\geq q^2 \brara \mL_{QW} \rho\kett=q^2 \brara \rho \kett, \label{eq:lower1}
\end{align}
respectively. 
These are equivalent to  
\[ q^2 \leq  \frac{\brara \mL_{p,\pi}\rho\kett}{\brara \rho \kett} \leq 1,\;\;  (\rho\neq 0)\]
which implies the first part. 
\item Second part:
Let us first consider the eigenequation $\mL_{p,\pi}\rho=\lambda \rho$  $(\rho\neq \bs{0})$. By taking the norm, we have  
\[ \brara \mL_{p,\pi}\rho \kett=|\lambda|^2\brara  \rho\kett. \]
Then (\ref{eq:lower1}) and (\ref{eq:upper1}) immediately imply 
\[q \leq |\lambda|\leq 1.\] 
This concludes 
\[ \mathrm{spec}(\mL_{p,\pi})\subseteq \{z\in\mathbb{C}\;|\;q\leq |z|\leq 1\}. \]
Then in the rest of the proof, let us characterize the eigenspace of the eigenvalues on the boundaries. 
\begin{enumerate}
\item $|\lambda|=q$ case: 
If $|\lambda|=q$, then the equality in (\ref{eq:lower1}) with $\rho\in\ker(\lambda-\mL_{p,\pi})$
attains. By (\ref{eq:lower}), we have $\brara \mL_\pi\rho \kett=0$ and $\mL_{p,\pi}\rho=q\mL_{QW}\rho=\lambda\rho$. 
Therefore, 
if $|\lambda|=q$, then 
\[\rho\in \ker\left(\frac{\lambda}{|\lambda|}-(\mL_{QW}-\mL_\pi)\right).\]

Conversely, if $\rho\in \ker(\mu-(\mL_{QW}-\mL_\pi))$ with some complex number $|\mu|=1$, then $\brara\rho\kett=\brara (\mL_{QW}-\mL_\pi)\rho \kett$. On the other hand, the unitarity of $\mL_{QW}$ implies \[\brara\rho\kett=\brara\mL_{QW}\rho\kett=\brara\mL_\pi\rho\kett+\brara (\mL_{QW}-\mL_\pi)\rho \kett.\]
Here we used Lemma~\ref{lem:1}. 
Then $\brara \mL_\pi\rho \kett$ must be $0$. 
This implies 
\[ \mL_{p,\pi}\rho=(q\mL_{QW}+p\mL_{\pi})\rho=(q\mu)\;\rho, \]
that is, $\rho\in \ker(\lambda-\mL_{p,\pi})$ with $|\lambda|=q$. 
Then for any $\lambda\in \mathrm{spec}(\mL_{p,\pi})$ with $|\lambda|=q$, we have 
\[ \ker(\lambda-\mL_{p,\pi})=\ker\left(\frac{\lambda}{|\lambda|}-(\mL_{QW}-\mL_\pi)\right). \]
\item $|\lambda|=1$ case:  
If $|\lambda|=1$, then the equality in (\ref{eq:upper1}) with $\rho\in\ker(\lambda-\mL_{p,\pi})$ attains. By (\ref{eq:upper}), we have $\brara (\mL_{QW}-\mL_{\pi})\rho \kett=0$. 
This implies $\mL_{\pi}\rho =\lambda \rho$. 
Therefore, if $|\lambda|=1$, then 
\[ \rho\in\ker(\lambda-\mathcal{L}_\pi). \]
Conversely, if $\rho\in \ker(\lambda-\mL_\pi)$ with $|\lambda|=1$, then $\brara\rho\kett=\brara \mL_\pi\rho\kett$. On the other hand, the unitarity of $\mL_{QW}$ implies 
\[ \brara\rho\kett=\brara \mL_{QW}\rho \kett=\brara \mL_{\pi}\rho \kett+\brara (\mL_{QW}-\mL_\pi)\rho \kett. \]
 Here we used Lemma~\ref{lem:1} in the last equality. Then $\brara (\mL_{QW}-\mL_\pi)\rho \kett$ must be $0$. This implies  
\begin{align*}
\mL_{p,\pi}\rho &= (\;\mL_{\pi}+q(\mL_{QW}-\mL_\pi)\;)\rho \\
&= \lambda \rho.
\end{align*}
Then we conclude $\ker(\lambda-\mL_{p,\pi})=\ker(\lambda-\mL_\pi)$. 
\end{enumerate}
Finally, let us see that $\sum_{|\lambda|=1}\ker(\lambda-\mathcal{L}_\pi)\neq \{\bs{0}\}$ by checking that 
$\rho(a,b)=\delta_{a,b} $ for any $a,b\in A$, that is, $\rho=I_A$ belongs to $\ker(1-\mL_{\pi})$ as follows:
\begin{align*}  \mathcal{L}_{\pi}\rho
&= \sum_{h}\Pi_h U\rho U^*\Pi_h = \sum_{h}\Pi_h = I_A =\rho,
\end{align*}
which implies that $\sum_{|\lambda|=1}\ker(\lambda-\mL_{p,\pi})\supseteq\{I_A\}\neq\{\bs{0}\}$.
\end{itemize}
\end{proof}
\begin{remark}
By Lemma~\ref{lem:1} and Theorem~\ref{prop:1}, 
another expression for the eigenspace of the eigenvalue on the boundary is described as follows: 
\begin{itemize}
\item For $|\lambda|=1$, 
\[ \ker(\lambda-\mL_{p,\pi})=\ker(\lambda-\mL_{QW})\cap \ker(\mL_{QW}-\mL_{\pi}) \]
\item For $|\lambda|=q$, 
\[ \ker(\lambda-\mL_{p,\pi})=\ker\left(\frac{\lambda}{|\lambda|}-\mL_{QW}\right)\cap \ker(\mL_{\pi}) \]
\end{itemize}
\end{remark}
The generalized eigenspace of $\lambda\in \mathrm{spec}(\mL_{p,\pi})$ is denoted by $\ker(\lambda-\mL_{p,\pi})^{m_\lambda}\subset \mathcal{H}$ with $\ker(\lambda-\mL_{p,\pi})^{m_\lambda+1}=\mathcal{H}$ for some $m_\lambda\geq 1$.  
In particular, if $m_\lambda=1$, then $\lambda$ is called {\it semi-simple}, in this paper. 
The following property is useful to detect the semi-simplicity of $\lambda$, which directly derives from the definition of the semi-simplicity:  
\begin{center}
``$(\lambda-\mathcal{L}_{p,\pi})^2\rho= 0$ with $\rho\neq 0$ implies $(\lambda-\mathcal{L}_{p,\pi})\rho=0$" $\Rightarrow$ ``$\mathcal{L}$ is semi-simple at $\lambda$". \;\;$(\star)$
\end{center}

Set the unit circle on the complex plane $\delta \mathbb{D}:=\{z\in \mathbb{C}\;|\;|z|=1\}$. 
Now we show the semi-simplicity of the eigenvalues on at least $\delta\mathbb{D}$ in the following proposition. 
\begin{proposition}[Semi-simplicity on $\delta \mathbb{D}$]\label{prop:semi-simple}
For any $\lambda\in \mathrm{spec}(\mathcal{L}_{p,\pi})\cap \delta \mathbb{D}$, $\lambda$ is semi-simple. 
\end{proposition}
\begin{proof}
Assume there exists $\rho\in \mathcal{H}$ such that 
\begin{equation}\label{eq:2jou} 
(\lambda-\mathcal{L})^2\rho=0. 
\end{equation}
with $|\lambda|=1$. 
By ($\star$), 
it is enough to show $\rho\in \ker(\lambda-\mathcal{L})$ . 
By (\ref{eq:2jou}), we remark that there exists $f\in \ker(\lambda-\mathcal{L})$ such that 
\[ f=(\lambda-\mathcal{L})\rho. \]
Then by induction, we have 
\[ \mathcal{L}^n\rho=\lambda^n\rho+n\lambda^{n-1}f  \]
for any $n\geq 1$.
Taking the square norm to both sides, we have
\begin{align*}
\brara \mathcal{L}^n\rho \kett=\brara \rho \kett+n^2\brara f \kett+2n\mathrm{Re}\brara \lambda\rho,f \kett.
\end{align*}
Combining this with Theorem~\ref{prop:1} (1), we obtain
\begin{align*}
\brara f \kett+\frac{2}{n}\mathrm{Re}\brara \lambda\rho,f \kett\leq 0
\end{align*}
for arbitrary $n\geq 1$. 
This means that $\brara f \kett$ must be $0$, that is, $(\lambda-\mathcal{L)}\rho=0$, which has completed the proof.  
\end{proof}
\section{Absorption state}
\subsection{Absorption space and limit behavior}
The absorption space describes the limit behavior of our walk because the contribution of the eigenstate with eigenvalue $|\lambda|<1$ decreases exponentially, while the generalized eigenstate with the eigenvalue $|\lambda|=1$ survives in the long time limit. By Proposition~\ref{prop:semi-simple}, the eigenvalues on $\delta \mathbb{D}$ are semi-simple, that is, these generalized eigenstates are the eigenstates. Then, the absorption state can be defined as follows. 
\begin{definition}[Absorption space]
\[\mathcal{H}_{abs}:=\bigoplus_{|\lambda|=1}\ker(\lambda-\mL_{p,\pi}). \]
\end{definition}
The following theorem shows that the absorption space is reduced to that of the correlated walk and the interpolating walk is absorbed in $\mathcal{H}_{abs}$ for $t\to \infty$ in the following meaning. 
\begin{theorem}[Asymptotic behavior of interpolating walk]\label{thm:abs}
\noindent 
\begin{enumerate}
\item For $p\neq 0$, we have 
\[ \mathcal{H}_{abs}=\bigoplus_{|\lambda|=1}\ker(\lambda-\mL_{\pi}). \]
\item Set $\rho_t:=\mathcal{L}_{p,\pi}\rho_{t-1}$ with some initial state $\rho_0\in \mathcal{H}_G$. 
Set $\mathcal{P}_{abs}$ as the projection operator onto $\mathcal{H}_{abs}$. Then 
\[ \lim_{t\to\infty}\sup_{(a,b)\in A\times A} | \rho_t(a,b)-(\mathcal{P}_{abs}\rho_0)(a,b) |= 0. \]
\end{enumerate}
\end{theorem}
\begin{proof}
\noindent 
\begin{enumerate}
\item Theorem~\ref{prop:1} (2) directly implies the statement of the first part. 
\item Let $\mathcal{P}_\lambda$ be the eigenprojection of $\lambda$. 
The eigennilponent of $\lambda$ is denoted by $D_\lambda$ such that if $\lambda$ is semi-simple, $D_\lambda=0$, if $\lambda$ is not semi-simple,  $D_\lambda\neq 0,\dots,D_\lambda^{m_\lambda}\neq 0,D_\lambda^{m_\lambda+1}=0$. Then the $t$-th iteration of $\mL_{p,\pi}$ is expressed by
\[ \rho_t=\sum_{\lambda}\lambda^{t} \left(\mathcal{P}_\lambda+\sum_{j=1}^{m_\lambda}\binom{t}{j}\mathcal{D}_\lambda^j\right)\rho_0. \]
Since $\lambda\in\delta \mathbb{D}\cap \mathrm{spec}(\mL_{p,\pi})$ is semi-simple by Proposition~\ref{prop:semi-simple}, 
Theorem~\ref{prop:1} implies 
\[ \lim_{t\to\infty}\left(\rho_t(a,b)-\sum_{\lambda\in \delta \mathbb{D}} \lambda^t(\mathcal{P}_\lambda \rho_0)(a,b)\right)=0 \]
for any $(a,b)\in A\times A$. 
Since $\rho_t$ is represented by a finite dimensional matrix, we obtain the desired conclusion. 
\end{enumerate}
\end{proof}

Put $\mathrm{spec}(\mL_{\pi})\cap \delta\mathbb{D}=\{\lambda_1,\dots,\lambda_s\}$ which is a multiset with $s=\dim \mathcal{H}_{abs}$. 
Set $u_{j}^{(R)}\in \mathbb{C}^{A\times A}$ as the {\it right} eigenvector of $\lambda_j$ such that $\mL_{\pi}u_{j}^{(R)}=\lambda_j u_{j}^{(R)}$. 
We can set the {\it left} eigenvector of $\lambda_j$ satisfying  
$\mL_{\pi}^*v_{j}^{(L)}=\lambda_j^{-1} v_{j}^{(L)}$ and $\brara v_i^{(L)},u_j^{(R)} \kett=\delta_{ij}$. 
Then the eigenprojection of $\lambda\in  \mathrm{spec}(\mL_{p,\pi})\cap \delta\mathbb{D}$, $\mathcal{P}_\lambda$, is represented by 
\[ \mathcal{P}_\lambda\rho=\sum_{j: \lambda=\lambda_j} \brara v_j^{(L)},\rho \kett u_j^{(R)} \]
for any $\rho\in \mathbb{C}^{A\times A}$ and $p\neq 0$ by Proposition~\ref{prop:1}(2).
We note that  Theorem~\ref{thm:abs} implies that 
\begin{equation} 
\rho_t\sim \sum_{j=1}^s \lambda_j^t\; \brara v_j^{(L)}, \rho_0\kett \;u_j^{(R)} 
\end{equation}
for large $t$. 
Here $\nu_t\sim \mu_t$ means $\lim_{t\to\infty} \sup_{(a,b)\in A\times A}|\nu_t(a,b)/\mu_t(a,b)|=1$. 
\subsection{Absorption state for $\pi=\pi_{max}$ case}
Put $T:=U\circ \bar{U}$, which is the intrinsic random walk on arcs of a graph $G=(X,A)$. Let us describe the directed graph $G_T:=(X_T,A_T)$ induced by $T$ as follows:  
\[X_T=A,\; (a,b)\in A_T \text{ if and only if } (T)_{a,b}\neq 0. \] 
Then $T$ is also the transition operator on vertices $X_T$ of the directed graph $G_T=(X_T,A_T)$. 
\begin{remark}
The geometry of dynamics by the unitary operator $U$ on $G$ is
different from that of $G$.
For example, let us consider the cycle $C_N$ with length $N \geq 1$, as the graph $G$, and
the time evolution operator of the Grover walk $U$ on $C_N$.
The arc set of $C_N$ is represented by $A(C_N) = \{ a_{i}, \bar{a}_{i}\, ; i\in \mathbb{Z}_N\}$
such that $ t(a_{i}) = o(a_{i+1})$ and $t(\bar{a}_{i+1}) = o(\bar{a}_{i})$.
Since the quantum coin assigned at each vertex $x$ is given by \[C_x=\begin{bmatrix}0 & 1 \\ 1 & 0\end{bmatrix},\] 
we have 
$(U)_{a_{i+1},a_{i}}=(U)_{\bar{a}_i, \bar{a}_{i+1}} = 1$ but
$(U)_{\bar{a}_i, a_i}=(U)_{a_i, \bar{a}_i} = 0$.
Then $G_{T}$ defined above has the $2$ connected components that are formed as the sequences of vertices $(a_0,a_1,\dots,a_{N-1})$ and $(\bar{a}_{N-1},\bar{a}_{N-2},\dots,\bar{a_0})$, respectively.
\end{remark}
Then, we set the number of connected components by $\omega(G_T)$. The transition matrix of the intrinsic random walk $T$ can be decomposed into $T=T_1\oplus \cdots \oplus T_{\omega(G_T)}$ under the decomposition of $G_T=G_1\sqcup \cdots \sqcup G_{\omega(G_T)}$. 
Here $G_j$ is a connected component of $G_T$. 
A useful sufficient condition for the connectivity of $G_T$ is as follows. 
\begin{remark}\label{rem:sufficient}
If all entries in the local coin matrix $C_x$ in (\ref{eq:QW}) are non-zero for any $x\in X$, then $G_T$ is connected; that is, $\omega(G_T)=1$. 
This is a natural situation since $G$ is strongly connected in the
sense of digraphs.
\end{remark}

A cycle with length $r\geq 2$ in $G_T$ is denoted by a sequence of arcs by 
$(\bs{a}_0,\bs{a}_1,\dots,\bs{a}_{r-1})$, 
where, $t(\bs{a}_j)=o(\bs{a}_{j+1})$ for any $j\in \mathbb{Z}_r$, and $o(\bs{a}_0),o(\bs{a}_1),\dots,o(\bs{a}_{r-1})$ are distinct. 
The period of the connected component $G_j$, $\tau(G_j)$, is defined by the greatest common divisor of the lengths of the cycles in $G_j$. 
For example, both $G_1$ and $G_2$ in $G_T$ of the Grover walk on $C_N$ have the period $N$. 
The Birkoff theorem~(e.g.,\cite[Theorem~8.7.2]{HJ}) ensures that every vertex in $G_T$ lies on some directed cycle due to the doubly stochastic property.
Assume $\tau(G_j)=\tau$. Then the vertex set of $G_j$ can be decomposed into $X_j= X_j^{(0)}\sqcup \cdots\sqcup X_j^{(\tau-1)}$ so that for any arc $\bs{a}$ in $G_j$, there uniquely exists $k\in \mathbb{Z}_\tau$ such that $t(\bs{a})\in X_{j}^{(k)}$ and $o(\bs{a})\in X_{j}^{(k-1)}$. 
\begin{remark}
If all entries in the local coin matrix $C_x$ in (\ref{eq:QW}) are non-zero for any $x\in X$, then 
$\tau (G) = 2 \text{\  or\ } 1$ when $G$ is bipartite or not, respectively, since $G$ is a symmetric digraph.
\end{remark}
Moreover, the doubly stochastic property also implies the following lemma, which
gives a property of geometry of dynamics on $G$ with $U$.
\begin{lemma} Let $X_j^{(0)},\dots,X_j^{(\tau-1)}$ be defined as above. Then we have 
\[|X_j^{(0)}|=\cdots=|X_j^{(\tau-1)}|=|X_j|/\tau.\] 
\end{lemma}
\begin{proof}
The transition matrix $T$ can be decomposed into $\omega(G_T)$ matrices such that $T=\oplus_{j=1}^{\omega(G_T)}T_j$. 
From now on,  we concentrate on the connected component labeled by $j$. Assume $T_j$ has a period $\tau_j$. 
Set $T_j^{(s+1,s)}$ $(s\in\mathbb{Z}_\tau)$ as the submatrix of $T_j$ which is the $|X_j^{(s+1)}|\times |X_j^{(s)}|$ matrix such that 
\[ T_j^{(s+1,s)}=\{(T_j)_{\ell,m}\}_{\ell\in X_j^{(s+1)},\;m\in X_j^{(s)}}. \]
We note that $T_j$ is described by the following block partitioning: 
\[ T_j=\begin{bmatrix} 
0 & 0 & \cdots & 0 & T_j^{(0,\tau-1)}\\
T_j^{(1,0)} & 0 & \cdots & 0 & 0\\
0 & T_j^{(2,1)} & \ddots & \vdots& \vdots\\
 \vdots&\vdots & \ddots & 0&  0\\
 0&0 &\cdots & T_j^{(\tau-1,\tau-2)}& 0
\end{bmatrix}. \]
Put the all-ones vector of size $|X_j^{(k)}|$ by $\bs{1}_{k}$. 
Noting $T_j$ is also doubly stochastic,  
we have 
\[ T_j^{(s+1,s)}\bs{1}_{s}=\bs{1}_{s+1},\;\;\left(T_j^{(s+1,s)}\right)^*\bs{1}_{s+1}=\bs{1}_{s}. \]
This implies 
\begin{align*}
|X_j^{(s)}|&=\bra \bs{1}_s,\bs{1}_s \ket \\
 &= \langle \bs{1}_{s}, (T_j^{(s+1,s)})^*\bs{1}_{s+1}\rangle \\
 &=  \langle T_j^{(s+1,s)}\bs{1}_{s},\bs{1}_{s+1}\rangle\\
 &=  \langle \bs{1}_{s+1},\bs{1}_{s+1}\rangle \\
 &= |X_j^{(s+1)}|
\end{align*}
for any $s\in\mathbb{Z}_\tau$. 
\end{proof}
For $\zeta\in \mathbb{C}$ with $\zeta^\tau=1$, set $\bs{\zeta}_j \in \mathbb{C}^A$ $(j=1,\dots,\omega(G_T))$ such that 
\[ \bs{\zeta}_j(a)=\begin{cases}
\zeta^0 & \text{: $a\in X_j^{(0)}$} \\
\zeta^{-1} & \text{: $a\in X_j^{(1)}$} \\
\vdots \\
\zeta^{-(\tau-1)} & \text{: $a\in X_j^{(\tau-1)}$}\\
0 & \text{: otherwise}
\end{cases} \]
\begin{fact}[\cite{HJ}]\label{fact:1}
Let $T=U\circ \bar{U}$ and $G_T$ has the connected components $\{G_1,\dots,
G_{\omega(G)}\}$. Then 
\[ \mathrm{spec}(T)\cap \delta \mathbb{D}=\bigcup_{j=1}^{\omega(G_T)} \{\zeta\in\mathbb{C}\;|\; \zeta^{\tau(G_j)}=1\}, \]
which are simple, 
and for any $\zeta\in \mathrm{spec}(T)\cap \delta \mathbb{D}$, 
\[ \ker(\zeta-T)=\bigoplus_{j:\;\zeta^{\tau(G_j)}=1}\mathrm{span}\{\bs{\zeta}_j\}. \]
\end{fact}
\begin{remark}\label{rem:projT}
Since $T$ is doubly stochastic, 
the {\it left} eigenvector of the eigenvalue $\zeta\in \mathrm{spec}(T)\cap \delta \mathbb{D}$ corresponding to the right eigenvector $\bs{\zeta}_j$ is described by \[\bs{\zeta}_j^*/|X_j|.\] 
Then, the eigenprojection $\Pi$ of $\zeta$ restricted to $T_j$ is expressed by 
\[ \Pi f=\frac{1}{|X_j|} \langle\bs{\zeta}_j,f\rangle \;\bs{\zeta}_j \]
for any $f\in \mathbb{C}^{X_{j}}$. 
\end{remark}
Note that if the diagonal elements of any local coin matrix $C_x$ $(x\in X)$ are non-zero, then $(T)_{a,\bar{a}},\;(T)_{\bar{a},a} \neq 0$ for any $a\in X_T$, which implies that there are cycles with period $2$ in $G_T$. Then we have the following statement. 
\begin{remark}\label{rem:period}
If all entries in the local coin matrix $C_x$ are non-zero for any vertices $x\in X$, then $G_T$ is connected, and the period $\tau(G_T)\in \{1,2\}$. 
Moreover, 
\[ \mathrm{spec}(T)\cap \delta\mathbb{D}=\begin{cases}\{\pm 1\} & \text{: $G_T$ is bipartite, that is, $\tau(G_T)=2$,}\\ \{+1\} & \text{: $G_T$ is non-bipartite, that is, $\tau(G_T)=1$}\end{cases} \]
and 
\begin{align*}
\ker(1-T) &= \mathrm{span}\{\bs{1}_A\}, \\
\ker(1+T) &= \mathrm{span}\{\bs{1}_A^\flat\}\;\;(\text{if $G_T$ is bipartite with the partite set $X=X_0\sqcup X_1$}).
\end{align*}
Here,  for any $a\in A$, 
\[\bs{1}_A(a)=1 \text{ and }\bs{1}_A^\flat(a)=\begin{cases}1 & \text{: $a\in X_0$,}\\ -1 & \text{: $a\in X_1$. }\end{cases}
\]

\end{remark}
To state the following theorem, let us prepare the following notations. 
\begin{itemize}
\item $\sigma_j: \mathbb{C}^{\{0,1,\dots,\tau(G_j)-1\}}\to \mathbb{C}^{\{0,1,\dots,\tau(G_j)-1\}}$ as the permutation matrix such that 
\[ (\sigma_jf)(1)=f(0),\;(\sigma_jf)(2)=f(1),\dots,(\sigma_jf)(0)=f(\tau(G_j)-1)  \]
for any $f\in \mathbb{C}^{\{0,1,\dots,\tau(G_j)-1\}}$. 
\item $\bs{r}_j\in \mathbb{C}^{\{0,1,\dots,\tau(G_j)-1\}}$ as the ratio of the initial state in $X_j^{(\kappa)}$ such that 
\[ \bs{r}_j(\kappa)=\sum_{a\in X_j^{(\kappa)}}\rho_0(a,a) \]
for any $\kappa\in\{0,1,\dots,\tau(G_j)-1\}$. 
\item $\bs{1}_j\in \mathbb{C}^{\{1,2,\dots,|X_j|/\tau(G_j)\}}$ as the all-ones vector such that 
\[ \bs{1}_j(x)=1\]
for any $x\in \{1,2,\dots,|X_j|/\tau(G_j)\}$. 
\item $\Phi_j^*:\mathbb{C}^{X_j}\to \mathbb{C}^{A\times A}$ ($j\in \{1,2,\dots,s\}$) such that 
\[ (\Phi_j^*f)(a,b)=\delta_{a,b}\begin{cases}f(a) & \text{: $a\in X_j$}\\ 0 & \text{: otherwise}\end{cases} \]
for any $a,b\in A$. 
\end{itemize}
Now we are ready to state our main theorem. 
\begin{theorem}\label{thm:max}
Assume  
$G_T$ has $s:=\omega(G_T)$ connected components 
\[G_1=(X_1,A_1),\dots,G_{s}=(X_{s},A_{s}). \]  
For the maximal partition $\pi_{max}$ and the parameter $p\neq 0$, divide the density matrix $\rho_t=\mathcal{L}_{p,\pi}\rho_{t-1}\in \mathbb{C^{A\times A}}$ ($t\geq 1$) by $A=X_1\sqcup \cdots \sqcup X_s$ such that    
\[ \rho_t=\{\; \rho_t|_{X_\ell,X_m}\;\}_{\ell,m=1}^s, \] 
where 
\[ \rho_t|_{X_\ell,X_m} (a,b)=\begin{cases} \rho_t(a,b) & \text{: $a\in X_\ell,\;b\in X_m$,}\\
0 & \text{: otherwise.}\end{cases}\]
Then we have
\[ \lim_{N\to\infty}\rho_{\tau(G_\ell)N+\kappa}|_{X_\ell,X_m}=\delta_{\ell,m}\frac{\tau(G_\ell)}{|X_\ell|} \Phi^*_\ell \left(\;(\sigma_\ell^\kappa \bs{r}_\ell)\otimes \bs{1}_\ell\;\right) \]
for any $\kappa\in \{0,1,\dots, \tau(G_j)-1\}$ and $\ell,m\in\{1,2,\dots,s\}$. 
\end{theorem}
\begin{remark}
The absorption behavior of $\rho_t|_{X_\ell,X_\ell}$ described by Theorem~\ref{thm:max} is expressed by 
\begin{multline*}
(\star)\mapsto \frac{\tau(G_\ell)}{|X_\ell|}\mathrm{diag}[r_0\cdots r_0\; r_1\cdots r_1\;\cdots \;r_{\tau-1}\cdots r_{\tau-1}]\\
\mapsto \frac{\tau(G_\ell)}{|X_\ell|}\mathrm{diag}[r_{\tau-1}\cdots r_{\tau-1}\; r_0\cdots r_0\;\cdots \;r_{\tau-2}\cdots r_{\tau-2}]
\mapsto \cdots \\
\cdots \mapsto \frac{\tau(G_\ell)}{|X_\ell|}\mathrm{diag}[r_1\cdots r_1\; r_2\cdots r_2\;\cdots \;r_{0}\cdots r_{0}]\mapsto(\star)
\end{multline*}
Here $\mathrm{diag}[a_1\; a_2\;\cdots\;a_{|X_\ell|}]$ is the diagonal matrix whose diagonal element of $j$ is $a_j$ ($j=1,2,\dots,|X_\ell|$). 
\end{remark}
\begin{proof}
By Theorem~\ref{thm:abs}, the absorption state of $\mL_{\pi_{max},p}^t\rho_0$ coincides with that of $\mL_{\pi_{max},1}\rho_0$. 
Then, it is enough to find the absorption state of the correlated walk $\mL_{\pi_{max}}=\mL_{\pi_{max,1}}$. 
Let $\rho_t\in \mathcal{H}_G$ be the $t$-th iteration of the correlated walk with the maximal decomposition $\pi=\pi_{max}$. 
Put the probability distribution of the intrinsic random walk by $\nu_t\in\mathbb{C}^A$ 
such that 
\begin{align*}
\nu_0(a) &= \rho_0(a,a) \;\text{(for any $a\in A$),}\\
\nu_t &= T\nu_{t-1}\;\;(t\geq 1). 
\end{align*}
By Remark~\ref{rem:CRW}, we have 
\[ \rho_t=\sum_{j=1}^{\omega(G_T)}\Phi_j^* \nu_t|_{X_j}  \]
for any $t\geq 1$.  
Here $\nu_t|_{X_j}$ is the restriction of $\nu_t$ to $\mathrm{span}\{\delta_a\;|\;a\in X_j\}\subset \mathbb{C}^A$. 
The asymptotic behavior of the random walk is governed by the eigenspace whose eigenvalues lie on $\delta \mathbb{D}$.  
Then, by Fact~\ref{fact:1} and Remark~\ref{rem:projT}, 
\begin{align}\label{eq:nu}
\nu_t\sim 
\sum_{j=1}^{\omega(T_G)}\left( \sum_{\zeta\;:\;\zeta^{\tau(G_j)}=1}\zeta^t \frac{\bra \bs{\zeta}_j,\;\nu_0 \ket }{|X_j|}\bs{\zeta}_j  \right).
\end{align}
We should remark that the shape ``$\sum_{\zeta\;:\;\zeta^{\tau(G_j)}=1}\zeta^t \frac{\bra \bs{\zeta}_j,\;\cdot \ket }{|X_j|}\bs{\zeta}_j$" forms the spectral decomposition of $\frac{\tau(G_j)}{|X_j|}\sigma_j^t\otimes J_j$, where $J_j$ is the all-ones matrix of size $|X_j|/\tau(G_j)$. 
The initial state $\nu_0$ can be rewritten as  $\nu_0=\sum_{s=0}^{\tau(G_j)-1}(\delta_s\otimes \nu_0^{(s)})$ by using $\nu_0^{(s)}\in \mathbb{C}^{X_j^{(s)}}$'s satisfying $\langle \bs{1}_j,\nu_0^{(s)}\rangle=\bs{r}_j(s)$. 
Then we have 
\begin{align*} 
\nu_t|_{X_j} &\sim \frac{\tau(G_j)}{|X_j|}\sum_{s=0}^{\tau(G_j)-1}\delta_{s+t}\otimes  \bs{r}_j(s)\bs{1}_j \\
&= \frac{\tau(G_j)}{|X_j|}(\sigma_j^t \bs{r}_j)\otimes \bs{1}_j,
\end{align*}
which implies 
\[ \nu_{\tau(G_j) N+\kappa}|_{X_j}\sim \frac{\tau(G_j)}{|X_j|}(\sigma_j^\kappa \bs{r}_j)\otimes \bs{1}_j, \]
where the symbol \(\otimes \) denotes the Kronecker product. 
By operating $\Phi_j^*$ on both sides, we obtain the desired conclusion. 
\end{proof}
\begin{corollary}\label{cor:bi}
Assume all entries in the local coin matrix $C_x$ are non-zero for any vertex $x\in X$. 
Set the partition $\pi$ as the maximal partition $\pi_{max}$.
Let $\rho_t$ be the $t$-th iteration of the interpolating walk with the initial state $\rho_0$ and the parameter $p\neq 0$.
Then, we have 
\begin{itemize}
\item if $G$ is non-bipartite, 
\[ \lim_{t\to\infty}\rho_t =  \frac{1}{|A|}I_A, \]
\item if $G$ is bipartite with the partite set $X=X^{(0)}\sqcup X^{(1)}$, for $j\in\{0,1\}$, 
\[ \lim_{s\to\infty}\rho_{2s+j}=\frac{2}{|A|} \begin{bmatrix}r_j I_{X^{(0)}} & 0 \\ 0 & r_{\neg j} I_{X^{(1)}}\end{bmatrix},  \]
where $\neg 0=1$ and $\neg 1=0$.
\end{itemize}
\end{corollary}
\begin{proof}
By Remark~\ref{rem:period}, we have $\omega(G_T)=1$ and $\tau(G_T)\in\{1,2\}$.  
Note that $\tau(G_T)=2$ if and only if $G$ is bipartite. 
Inserting $\omega(G_T)=1$ and $\tau(G_T)\in\{1,2\}$ into the statement of Theorem~\ref{thm:max}, we obtain the desired conclusion. 
\end{proof}
\begin{corollary}\label{cor:maxlimit}
Assume all entries in the local coin matrix $C_x$ are non-zero for any vertex $x\in X$.
Set the partition $\pi$ as the maximal partition $\pi_{max}$.
Let $\mu_t$ be the distribution at time $t$ for the interpolating walk with the initial state $\rho_0$ and the parameter $p\neq 0$. \\
If $G$ is non-bipartite, 
\[ \lim_{t\to\infty}\mu_{t}(a)= \frac{1}{|A|}, \]
and if $G$ is bipartite with the partite set $X=X^{(0)}\sqcup X^{(1)}$, 
\[ \lim_{s\to\infty}\mu_{2s+j}(a)= \frac{2}{|A|} \begin{cases}r_j & \text{: $t(a)\in X^{(0)}$,}\\ r_{\neg j} & \text{: $t(a)\in X^{(1)}$.} \end{cases} \]
\end{corollary}
\begin{proof}
Since $\mu_t(a)=\rho_t(a,a)$,  Corollary~\ref{cor:bi} immediately implies the desired conclusion. 
\end{proof}
\subsection{Absorption state for $\pi=\pi_c$ case}
In this subsection, we discuss the absorption state for $\pi=\pi_c$, which is an intermediate partition between $\pi_o$ and $\pi_{max}$ as discussed in Section~\ref{sect:pa}.
Let $G=(X,A)$ be $d$-arc colorable and set the color $H=\{1,2,\dots,d\}$ with the coloring by $(\gamma,\phi)$. 
Under the partition
\[ \pi_c: A=\bigsqcup_{h\in H} \gamma^{-1}(h)=\bigsqcup_{h\in H} (X\times \{h\}), \]
the operator $\mathcal{L}_{\pi_c}$ is defined by 
\[ \mathcal{L}_{\pi_c}\rho=\sum_{h\in H}\Pi_hU\rho U^*\Pi_h. \]
Here $\Pi_h$ is the projection onto arcs colored by $h$, which is defined in (\ref{eq:proh}).
In this subsection,  we set the following assumptions for simplicity.
\begin{assumption}\label{ass:1} 
Let $G=(X,A)$ be $d$-arc colorable with the coloring $(\gamma,\phi)$. 
\begin{enumerate}
\item
Set $C_0$ as a unitary matrix on $\mathbb{C}^H$. 
Under the isomorphism $\{a\in A \;|\; t(a)=x\}\cong H$ with the one-to-one correspondence $a\leftrightarrow\gamma(a)$, every
local coin matrix $C_x$ acting on $\mathbb{C}^{\{a\in A \;|\; t(a)=x\}}$ is isomorphic to $C_0$ acting on $\mathbb{C}^H$.   

\item All the elements of $C_0$ are non-zero. 
\end{enumerate}
\end{assumption}
\noindent 

For a countable set $\Omega\in\{X, H, A\}$, we put the standard basis of $\mathbb{C}^\Omega$ as $\{|\omega \rangle\}_{\omega\in \Omega}$. 
Here $\bra \omega|:=(|\omega\ket)^*$. 
Then $|\omega \ket\bra \omega|$ is the projection onto the standard basis $|\omega\ket$, while $\bra \omega|\omega\ket$ is the inner product, which becomes $1$. 

Recall that since $G$ is $d$-arc colorable, $A\cong X\times H$.
Using this one-to-one correspondence, we can identify each arc $a\in A$ with the pair of $(t(a);\gamma(a))$.
Then we have 
\[\mathbb{C}^{A}\cong (\mathbb{C}^{H})^X\cong \mathbb{C}^X\otimes \mathbb{C}^H.\] 
Here, the bijection between them is represented by the following unitary operator $ \mathcal{U}_1:\mathbb{C}^A\to (\mathbb{C}^H)^X$.
\[ (\mathcal{U}_1\psi)(x)=\begin{bmatrix} \psi(x;1) & \psi(x;2) & \cdots & \psi(x;d) \end{bmatrix}^\top\in \mathbb{C}^H \]
for any $\psi\in \mathbb{C}^A$ and $x\in X$. 
Here $(x;h)$ presents the arc $a$ with $t(a)=x$ and $\gamma(a)=h$. 
The adjoint $\mathcal{U}_1^*: (\mathbb{C}^H)^X\to \mathbb{C}^A$ is given by 
\[ (\mathcal{U}_1^*f)(a)=\langle\; \gamma(a)\;|\;f(\;t(a)\;)\;\rangle \]
for any $f\in (\mathbb{C}^H)^X$ and $a\in A$.  
We note that $\mathcal{U}_1\mathcal{U}_1^*=I_{(\mathbb{C}^H)^X}$ and $\mathcal{U}_1^*\mathcal{U}_1=I_{\mathbb{C}^A}$.
Using this unitary map, we define $\mathcal{U}_2:\mathbb{C}^{A\times A}\to (\mathbb{C}^{H\times H})^{X\times X}$  by 
\[ \mathcal{U}_2\rho=\mathcal{U}_1\;\rho\;\mathcal{U}_1^* \]
for any $\rho\in \mathbb{C}^{A\times A}$. Here, $\bs{\rho}:=\mathcal{U}_2\rho$ is regarded as a matrix valued function from $X\times X$ to $\mathbb{C}^{H\times H}$, that is, 
\[ \bs{\rho}(x,y)=\{\; \rho\left((x;h_1),(y;h_2)\right)\; \}_{h_1,h_2\in H}\in \mathbb{C}^{H\times H}. \]

We note that $\gamma(a)\neq \gamma(b)$ for any $a,b\in A$ with $o(a)=o(b)$ by the $d$-arc colorable conditions of (1) and (2) in Definition~\ref{def:arcc}. 
Thus for each color $h\in H$, the subset of arcs colored by $h\in H$ divides into the arc disjoint oriented $2$-factors. Such a subgraph can be represented by a permutation $\sigma_h: X\to X$, that is, for any $a\in A$ with $\gamma(a)=h$,  $\sigma_h(o(a))=t(a)$. 
Note that $\sigma_h^{-1}(x)=\sigma_{\phi(h)}(x)$.
The time evolution of the correlated walk on $d$-arc colorable graph can be simply described by the following form in $(\mathbb{C}^{H\times H})^{X\times X}$.  
\begin{lemma}\label{lem:Ph}
Set $P_h=|h\rangle \langle \phi(h)|\;C_0$ for any $h\in H$ and also set $\hat{\mL}_{\pi_c}:=\mathcal{U}_2\mL_{\pi_c}\mathcal{U}_2^*$. 
Then we have 
\begin{equation}\label{eq:OQRW}
(\hat{\mathcal{L}}_{\pi_c}\bs{\rho})(x,y)=\sum_{h\in H} P_h\bs{\rho}(\sigma_h^{-1}(x),\sigma_h^{-1}(y))P_h^*,
\end{equation}
for any $\bs{\rho}\in (\mathbb{C}^{H\times H})^{X\times X}$ and $x,y\in X$. 
\end{lemma}
\begin{proof}
First, we rewrite the time evolution operator $U$ with the constant coin $C_x=C_0$ $(x\in X)$ in the space of $(\mathbb{C}^H)^X$ as follows. 
The coin operator $C$ is rewritten by 
\begin{equation}\label{eq:coin}
(\mathcal{U}_1C\mathcal{U}_1^*f)(x)=C_0f(x).
\end{equation}
Next, let us represent the action of $S$ in $(\mathbb{C}^H)^X$. 
Recall that the shift operator $S:\mathbb{C}^A\to \mathbb{C}^A$ acts as $|a\rangle \stackrel{S}{\to} |\bar{a}\rangle$. 
The inverse arc represented by $(x;h)\in A$ is described by $(\;\sigma_h^{-1}(x),\;\phi(h)\;)\in A$.
Then the shift operator in $(\mathbb{C}^H)^X$ is represented by 
\begin{multline}\label{eq:shift}
(\mathcal{U}_1S\mathcal{U}_1^*f)(x)\\=\begin{bmatrix} \left\langle \;\phi(h_1)\;|\;f\left(\;\sigma_{h_1}^{-1}(x)\;\right)\;\right\rangle & 
\left\langle \;\phi(h_2)\;|\;f\left(\;\sigma_{h_2}^{-1}(x)\;\right)\;\right\rangle & \cdots & 
\left\langle \;\phi(h_d)\;|\;f\left(\;\sigma_{h_d}^{-1}(x)\;\right)\;\right\rangle \end{bmatrix}^\top 
\end{multline}
Combining (\ref{eq:coin}) with (\ref{eq:shift}), we have 
\begin{align*}
(\mathcal{U}_1U\mathcal{U}_1^*f)(x) &=  
\begin{bmatrix} \;\left\langle \;\phi(h_1)\;|\;\mathcal{U}_1C\mathcal{U}_1^*f\left(\;\sigma_{h_1}^{-1}(x)\;\right)\;\right\rangle & 
\cdots & 
\left\langle \;\phi(h_d)\;|\;\mathcal{U}_1C\mathcal{U}_1^*f\left(\;\sigma_{h_d}^{-1}(x)\;\right)\;\right\rangle\; \end{bmatrix}^\top \\
&= |h_1\rangle\langle \;\phi(h_1)\;|C_0f\left(\;\sigma_{h_1}^{-1}(x)\;\right)+\cdots+|h_d\rangle\langle \;\phi(h_d)\;|C_0f\left(\;\sigma_{h_d}^{-1}(x)\;\right) \\
&= P_1\; f\left(\;\sigma_{h_1}^{-1}(x)\;\right)
+\cdots+ P_d\; f\left(\;\sigma_{h_d}^{-1}(x)\;\right).
\end{align*}
Moreover,  
\[ (\mathcal{U}_1\Pi_{h}\mathcal{U}_1^*f)(x)=|h\rangle \langle h| f(x).  \]
Here $\Pi_h$ is the projection onto arcs colored by $h$, which is defined in (\ref{eq:proh}). 
Then, for any $h\in H$, we have 
\[ (\mathcal{U}_1\Pi_h U\mathcal{U}_1^*f)(x)= P_h\; f(\sigma_h^{-1}(x)). \]
In other words, since the standard basis of $(\mathbb{C}^H)^X$ is given by $\{ |x\ket\otimes|h\ket\;|\;x\in X,\;h\in H \}$, we have 
\[ \mathcal{U}_1\Pi_hU\mathcal{U}_1^*=\hat{\sigma}_h\otimes P_h. \]
Here, $\hat{\sigma}_h:\mathbb{C}^X\to \mathbb{C}^X$ is the matrix representation of $\sigma_h$ such that 
\[ (\hat{\sigma}_hf)(x)=f(\sigma_h^{-1}(x)), \]
for any $f\in \mathbb{C}^X$ and $x\in X$. 

Therefore, for any $x,y\in X$, we have 
\begin{align*}
(\mathcal{U}_2\mathcal{L}_{\pi_c}\mathcal{U}_2^*\bs{\rho})(x,y)
&=\sum_{h\in H}(\mathcal{U}_1\Pi_hU\mathcal{U}_1^*) \bs{\rho} (\mathcal{U}_1 \Pi_h U\mathcal{U}_1^*)^* (x,y) \\
&=\sum_{h\in H}\left((\hat{\sigma}_h\otimes P_h) \bs{\rho} (\hat{\sigma}_h\otimes P_h)^*\right) (x,y) \\
&=\sum_{h\in H}(\bra \sigma_h^{-1} (x)|\otimes P_h) \bs{\rho} (|\sigma_h^{-1}(y)\ket\otimes P_h^*) \\
&=\sum_{h\in H} P_h\bs{\rho}(\sigma_h^{-1}(x),\sigma_h^{-1}(y))P_h^*,
\end{align*}
for any $\bs{\rho}\in (\mathbb{C}^{H\times H})^{X\times X}$. 
\end{proof}
We give a connection to an {\it open quantum random walk}~\cite{AttalEtAl}.  
\begin{lemma}[Relation to an open quantum random walk]
Set the initial state $\rho_0$ so that $\supp(\rho_0)\subset \{ (x,x)\;:\;x\in X \}$, and put $\nu_t(x)=\rho_t(x,x)$ for $t=0,1,2,\dots$. Then we have  
\begin{equation}\label{eq:OQRW2}
\nu_t(x)=
(\mathcal{U}_2\mL_{\pi_c}\mathcal{U}_2^*\nu_{t-1})(x)=\sum_{h\in H}P_h\nu_{t-1}(\sigma_h^{-1}(x))P_h^*. \;\;(t\geq 1)
\end{equation}
\end{lemma}
\begin{proof}
We note that if $\supp(\rho)\subset \{ (x,x)\;:\;x\in X \}$, then $\supp(\mL_{\pi_c}\rho)\subset D$. 
Thus $\supp(\rho_t)\subset \{ (x,x)\;:\;x\in X \}$ for any $t\geq 0$. Then, applying Lemma~\ref{lem:Ph}, we obtain the desired conclusion. 
\end{proof}
\noindent Note that $\sum_{h}P_h^*P_h=I_H$ holds. 
This means that when we restrict the initial state $\bs{\rho}_0$ so that $\supp(\bs{\rho}_0)=\{(x,x)\;|\;x\in X\}$, the correlated walk with $\pi_c$ reproduces a time evolution of an open quantum random walk~\cite{AttalEtAl} on $d$-arc colorable graphs. 
Interesting constructions of open quantum walks living on the arcs (which are not necessarily $d$-arc colorable)
and on the vertices of connected graphs are discussed and provide a connection to a Markov chain in \cite{Joye}. 
\\

Let $\Sigma$ be the group generated by the products of $\sigma_h\otimes \sigma_h$'s, such that 
$\Sigma=\bra\sigma_h\otimes \sigma_h\ket_{h\in H}$.
For any $\bs{x}:=(x_1,x_2), \bs{y}:=(y_1,y_2)\in X\times X$, we define that 
$\bs{x}\stackrel{path}{\sim} \bs{y}$ if and only if there exists $\bs{\sigma}=\sigma\otimes \sigma\in \Sigma$ such that
$\bs{\sigma}(\bs{x})=\bs{y}$, that is, $y_1=\sigma(x_1)$ and $y_2=\sigma(x_2)$. 
The equivalent class is denoted by $(X\times X)/\stackrel{path}{\sim}$. 
For $B\in (X\times X)/\stackrel{path}{\sim}$, let us set the graph $G^{B}=(X^B,A^B)$ by
\begin{align}
&X^{B}=B(\subset X\times X) \notag\\
&a\in A^{B}(\subset A\times A) \Leftrightarrow \text{$\exists h\in H$ such that 
$t(a)=(\sigma_h\otimes \sigma_h)(o(a))$}. \label{eq:inducedgraph}
\end{align}
Note that $B_o:=\{(x,y)\in X\times X\;|\;(x_*,x_*)\stackrel{path}{\sim} (x,y)\}\cong X$ for any $x_*\in X$, and $G^{B_o}\cong G$.
See Figure~\ref{fig:B}. 
We note that $G^B$ is a symmetric digraph, that is, $a\in A^B$ if and only if $\bar{a}\in A^B$,  and  
also that even if the underlying graph $G$ is non-bipartite, there might be a $B\in (X\times X)/\stackrel{path}{\sim}$ such that $G^B$ is bipartite; see Figure~\ref{fig:B} (a).  
For any $B\in (X\times X)/\stackrel{path}{\sim}$,  set the projection onto $A^B$ such that 
\[ (\Pi_{B}\bs{\rho})(x,y)=\begin{cases} \bs{\rho}(x,y) & \text{: $(x,y)\in B$,}\\ 0 & \text{: otherwise.} \end{cases} \]
By the definition of $B$, we have 
\[ \hat{\mL}_{\pi_c}(\mathrm{Ran}\Pi_B)\subseteq \mathrm{Ran}\Pi_B. \]
This means that $\bs{\rho}_t$ is decomposed into $\oplus_{B} \bs{\rho}_t^{(B)}$ for any $t\geq 0$. 

On the other hand, the time evolution of $\hat{\mL}_{QW}:=\mathcal{U}\mL_{QW}\mathcal{U}^*$ can be expressed as
\[ (\hat{\mL}_{QW}\bs{\rho})(x,y)=\sum_{h,h'\in H}P_h \bs{\rho}(\sigma_h^{-1}(x),\sigma_{h'}^{-1}(y))P_{h'}^* \]
in a similar fashion to the case  $\mL_{\pi_c}$. 
Then, the time evolution operator of $\hat{\mL}_{\pi_c,p}:=\mathcal{U}\mL_{\pi_c,p}\mathcal{U}^*$ can be described by 
\[ \hat{\mL}_{\pi_c,p}=p\hat{\mL}_{\pi_c}+q\hat{\mL}_{QW}. \]
We note that $\mathrm{Ran}(\Pi_B)$ is not an invariant subspace for $\hat{\mL}_{\pi_c,p}$ ($p\neq 1$) in general. However, the following theorem shows that the absorption state of  $\bs{\rho}_{t,p}(\bs{x})$ depends only on which $\bs{x}$ belongs to $(X\times X)/\stackrel{path}{\sim}$, which is a consequence of Theorem~\ref{prop:1}. 
\begin{theorem}\label{prop:abs}
Suppose that $p>0$ and Assumption~\ref{ass:1} hold. 
Let $G$ be a $d$-arc colorable graph with the coloring $(\gamma,\phi)$. 
Let $\bs{\psi}_{t,p}$ be the $t$-th iteration of $\hat{\mL}_{\pi_c,p}$ such that $\bs{\rho}_{t+1,p}=\hat{\mL}_{\pi_c,p}\bs{\rho}_{t,p}$ with the initial state $\bs{\rho}_0$. 
For $B\in (X\times X)/\stackrel{path}{\sim}$, 
set $r^{(B)}=\sum_{\bs{x}\in X^B}\tr(\bs{\rho}_0(\bs{x}))$ if $G^B$ is non-bipartite, and $r^{(B)}_j=\sum_{\bs{x}\in X^B_j}\tr(\bs{\rho}_0(\bs{x}))$ $(j\in\{0,1\})$  if $G^B$ is bipartite with the partite set $X^B=X_0^B\sqcup X_1^B$. 
Then we have, for $\bs{x}\in X^B$, 
\begin{itemize}
\item if $G^B$ is non-bipartite, 
\[ \lim_{t\to\infty}\bs{\rho}_{t,p}(\bs{x})=\frac{r^{(B)}}{|A^B|}I_{\mathbb{C}^H}, \]
\item if $G^B$ is bipartite, for $\bs{x}\in X_j^{(B)}$, $j\in \{0,1\}$, 
\[ \lim_{s\to\infty}\bs{\rho}_{2s,p}(\bs{x})=\frac{2r^{(B)}_j}{|A^B|}I_{\mathbb{C}^H},\;\lim_{s\to\infty}\bs{\rho}_{2s+1,p}(\bs{x})=\frac{2r^{(B)}_{\neg j}}{|A^B|}I_{\mathbb{C}^H},  \]
\end{itemize}
where $\neg 0=1$ and $\neg 1=0$. 
\end{theorem}
\begin{proof}
By Theorem~\ref{prop:1}, the limit behavior of $\bs{\rho}_{t,p}$ coincides with that of $\bs{\rho}_{t,1}$, which is a purely correlated walk. 
Then it is enough to find the absorption state of $\bs{\rho}_{t,1}$. We put $\bs{\rho}_t:=\bs{\rho}_{t,1}$ and $\bs{\rho}_t^{(B)}:=\Pi_{B}\bs{\rho}_{t,1}$. Then, since $\mathrm{Ran}(\Pi_B)$ is invariant under the action of $\mL_{\pi_c}$, we have 
\[ \bs{\rho}_t^{(B)}=\hat{\mL}_{\pi_c}\bs{\rho}_{t-1}^{(B)}. \]
For $\bs{x}=(x,y)\in X\times X$, set $\bs{\sigma}_h$ by $\bs{\sigma}_h(\bs{x})=(\sigma_h(x),\sigma_h(y))$. 
Then, Lemma~\ref{lem:Ph} implies that  the time evolution restricted to $\mathrm{Ran}(\Pi_B)$ is described by  
\[ \bs{\rho}_t^{(B)}(\bs{x})=\sum_{h\in H}P_h \bs{\rho}_{t-1}^{(B)}(\bs{\sigma}_h^{-1}(\bs{x}))P_h^{*}. \]
Let us set 
$\bs{c}_\ell:=C_0^*|\phi(\ell)\rangle$ ($\ell\in H$), which is the column vector of $C^*_0$ labeled by $\ell$. 
The above equation implies the following: 
\begin{equation}\label{eq:eigeneqcorr}
\left\bra \bs{c}_\ell, \bs{\rho}^{(B)}_{t}(\bs{x})\bs{c}_m \right\ket =\sum_{h\in H} \left\bra\; \bs{c}_h,\;\bs{\rho}^{(B)}_{t-1}(\bs{\sigma}_h^{-1}(\bs{x}))\;\bs{c}_h\; \right\ket \;\bar{\bs{c}}_\ell(h)\;\bs{c}_m(h) \end{equation}
for any $\ell,m\in H$. 
In particular, putting $\eta^{(B)}_t(\bs{x},\ell):=\bra \bs{c}_\ell, \bs{\rho}^{(B)}_t(\bs{x})\bs{c}_\ell\ket$, we have 
\begin{align}\label{eq:eigeneqCayley}
&\eta_0^{(B)}(\bs{x},\ell)=\bs{\rho}_0(\bs{x},\ell), \notag \\ 
&\eta^{(B)}_{t}(\bs{x},\ell)=\sum_{h\in H}|(C_0)_{\phi(\ell),h}|^2 \eta^{(B)}_{t-1}(\bs{\sigma}_h^{-1}(\bs{x}),h),\; (t\geq 1)
\end{align}
for any $\bs{x}\in X^B$ and $h\in H$. 
Let us interpret that a pair $(\bs{x},\ell)\in X_B\times H$ corresponds to the arc in $A^B$ whose origin is $\bs{x}$ and terminus is $\bs{\sigma}_\ell(\bs{x})$. 
Then we have $X_B\times H\cong A^B$. 
From this correspondence, we find that the above recursion (\ref{eq:eigeneqCayley}) nothing but the time evolution of the intrinsic random walk on $G^B$, whose transition matrix is doubly stochastic.
Assumption~\ref{ass:1} (1) implies that  the periodicity $\tau(G^B)$ of the intrinsic random walk on $G^B$ is $1$ or $2$, that is, $G^B$ is non-bipartite ($\tau(G^B)=1$) or $G^B$ is bipartite ($\tau(G^B)=2$). 
In a similar fashion to the proof of   Corollary~\ref{cor:bi},  we obtain 
if $G^B$ is non-bipartite, 
\begin{align*}
\lim_{t\to\infty}\eta_t^{(B)}(\bs{x},\ell)= \frac{r_B}{|A^B|},
\end{align*}
and if $G^B$ is bipartite, for $\bs{x}\in X_h^B$ $j\in \{0,1\}$, 
\begin{align*}
\lim_{s\to\infty}\eta_{2s}^{(B)}(\bs{x},\ell)= \frac{2}{|A^B|}r_j^{(B)},\;
\lim_{s\to\infty}\eta_{2s+1}^{(B)}(\bs{x},\ell)= \frac{2}{|A^B|}r_{\neg j}^{(B)}.
\end{align*}
Then $\eta_t^{(B)}(\bs{x},h)=\bra \bs{u}_h,\bs{\rho}_t^{(B)}(\bs{x}) \bs{u}_h\ket$ is independent of $h\in H$, that is, the diagonal elements of $C_0\rho_t^{(B)}(\bs{x})C_0^*$ take the same non-zero value in the long time limit.  
Inserting this into (\ref{eq:eigeneqcorr}) and taking $t\to\infty$, we have 
\[ \lim_{t\to\infty}\bra \bs{c}_\ell,\bs{\rho}^{(B)}_t(\bs{x})\bs{c}_m \ket=0 \]
for any $\ell\neq m$, 
by the unitarity of $C^*$. 
This implies that $C_0\rho^{(B)}_t(\bs{x})C_0^*$ must be a diagonal matrix in the long time limit. Then, 
if $G^B$ is non-bipartite, 
\begin{equation}
\lim_{t\to\infty}\bs{\rho}_t^{(B)}(\bs{x})=\frac{r_B}{|A^B|}I_{\mathbb{C}^H},
\end{equation}
and if $G^B$ is bipartite, for any $\bs{x}\in X^B_j$ ($j\in\{0,1\}$)
\begin{align*}
\lim_{s\to\infty}\bs{\rho}_{2s}^{(B)}(\bs{x}) = \frac{2r_j^{(B)}}{|A^B|}I_{\mathbb{C}^H},\;\; 
\lim_{s\to\infty}\bs{\rho}_{2s+1}^{(B)}(\bs{x}) = \frac{2r_{\neg{j}}^{(B)}}{|A^B|}I_{\mathbb{C}^H},
\end{align*}
which is the desired conclusion. 
\end{proof}

Next, consider the correlated walk on the Cayley graph $G=(X,A)$ with $X=\bra H\ket$ with the generator $H$ such that if $h\in H$, then   $h^{-1}\in H$, and the identity element $e\notin H$. 
For any $g_1,g_2\in X$, $(g_1,g_2)\in A$ if and only if  there exists $h\in H$ such that $g_2=hg_1$. 
We remark that the Cayley graph is a special class of $|H|$-arc colorable graphs.
For fixed $b\in \bra H\ket$, 
define $D_b:=\{ (g_1,g_2)\;:\;g_2^{-1}g_1=b \}$. 
We note that 
\begin{itemize}
\item $D_b\cong \bra H \ket$; 
\item for $(g_1,g_2)\in D_b$, $(hg_1,hg_2)\in D_b$ for any $h\in H$. 
\end{itemize}
Then, we have 
\[ \bigoplus_{b\in\bra H \ket} D_b\cong \overbrace{\bra H\ket \oplus  \cdots \oplus \bra H\ket}^{|\bra H \ket|}\cong  (X\times X)/\stackrel{path}{\sim}   \]
Thus, the graph induced by $D_b$, $G^{D_b}$, is isomorphic to the original Cayley graph $G$ for any $b\in \bra H \ket$. 
See Fig.~\ref{fig:B} (b).

We put $\pi_c$ by $\pi_c: A=\sqcup_{h\in H} X\times\{h\} $. Lemma~\ref{lem:Ph} implies that the time evolution of the correlated walk induced by $\pi_c$ is 
\[ \bs{\rho}_{t+1}(g_1,g_2)=\sum_{h\in H}P_h\bs{\rho}_t(h^{-1}g_1,h^{-1}g_2)P_h^*. \]
\begin{corollary}
Suppose that $p>0$ and Assumption~\ref{ass:1} hold. 
Let $G=(X,A)$ be a Cayley graph with the generator $H$. 
For the initial state $\bs{\rho}_0$, 
set $r^{(b)}=\sum_{\bs{x}\in D_b}\tr(\bs{\rho}_0(\bs{x}))$ if $G$ is non-bipartite, and $r^{(b)}_j=\sum_{\bs{x}\in D^{(b)}_j}\tr(\bs{\rho}_0(\bs{x}))$ $(j\in\{0,1\})$  if $G$ is bipartite with the partite set $X=X_0\sqcup X_1$. 
Then we have, for $\bs{x}\in D^{(b)}$, 
\begin{itemize}
\item if $G$ is non-bipartite, 
\[ \lim_{t\to\infty}\bs{\rho}_{t,p}(\bs{x})=\frac{r^{(b)}}{|A|}I_{\mathbb{C}^H}; \]
\item if $G$ is bipartite, for $\bs{x}\in X_j^{(b)}$, $j\in \{0,1\}$, 
\[ \lim_{s\to\infty}\bs{\rho}_{2s,p}(\bs{x})=\frac{2r^{(b)}_j}{|A|}I_{\mathbb{C}^H},\;\lim_{s\to\infty}\bs{\rho}_{2s+1,p}(\bs{x})=\frac{2r^{(b)}_{\neg j}}{|A|}I_{\mathbb{C}^H}, \]
where $\neg 0=1$ and $\neg 1=0$. 
\end{itemize}
\end{corollary}
\begin{proof}
This follows from a direct application of  Theorem~\ref{prop:abs}. 
\end{proof}

Finally, we provide the stationary  distribution for the interpolating walk on $d$-arc colorable graphs. 
\begin{corollary}\label{cor:colorlimit}
Let $G=(X,A)$ be $d$-arc colorable. 
Suppose that $\pi=\pi_c$, $p>0$ and Assumption~\ref{ass:1} hold. 
Then we have, for the initial state $\bs{\rho}_0$, 
if $G$ is non-bipartite,  
\[ \lim_{t\to\infty}\mu_t(a)=\frac{1}{|A|}; \]
if $G$ is bipartite with the partite set $X=X_0\sqcup X_1$, for $j\in\{0,1\}$,
\[ \lim_{s\to\infty}\mu_{2s+j}(a)=\frac{2}{|A|}\begin{cases} r_j & \text{: $t(a)\in X_0$,}\\ r_{\neg j} & \text{: $t(a)\in X_1$,} \end{cases} \]
where $\neg 0=1$, $\neg 1=0$ and $r_j=\sum_{\bs{x}=(x,x),\;x\in X_j}\tr(\bs{\rho}_0(\bs{\it{x}}))$ $(j\in\{0,1\})$. 
\end{corollary}
\begin{proof}
Since $\mu_t(a)=\rho_t(a,a)=(\bs{\rho}_t(\bs{x}))_{\gamma(a),\gamma(a)}$ with $\bs{x}=(t(a),t(a))$ and 
$\sum_{a\in A}\rho_t(a,a)=1$,
the desired conclusion
immediately follows from 
Theorem~\ref{prop:abs}. 
\end{proof}

The walker moves within each connected component $G^B$ following the time evolution (\ref{eq:OQRW}), or it moves to other connected components with the thinning parameter $q(\hat{\mL}_{QW}-\hat{\mL}_{\pi_c})$. 
However, Theorem~\ref{thm:abs} shows that the walker in the absorption state never visits the other connected components. 
The ``ratio" of each connected component $G^B$ is expressed by $r_*^{B}$ in Theorem~\ref{prop:abs}. 

\begin{figure}[h]
\centering
\caption{$G^B$'s induced by arc colorings: The underling graphs in the absorption space of the  interpolating walks with arc colorings depicted in the left corner (a), (b) and (c) are the following disconnected graphs $G_B's$. 
The walker in the absorption state never visit to the other connected components.
}
\label{fig:B}
    \begin{minipage}{.6\linewidth}
    \centering
    \includegraphics[width=1\columnwidth]{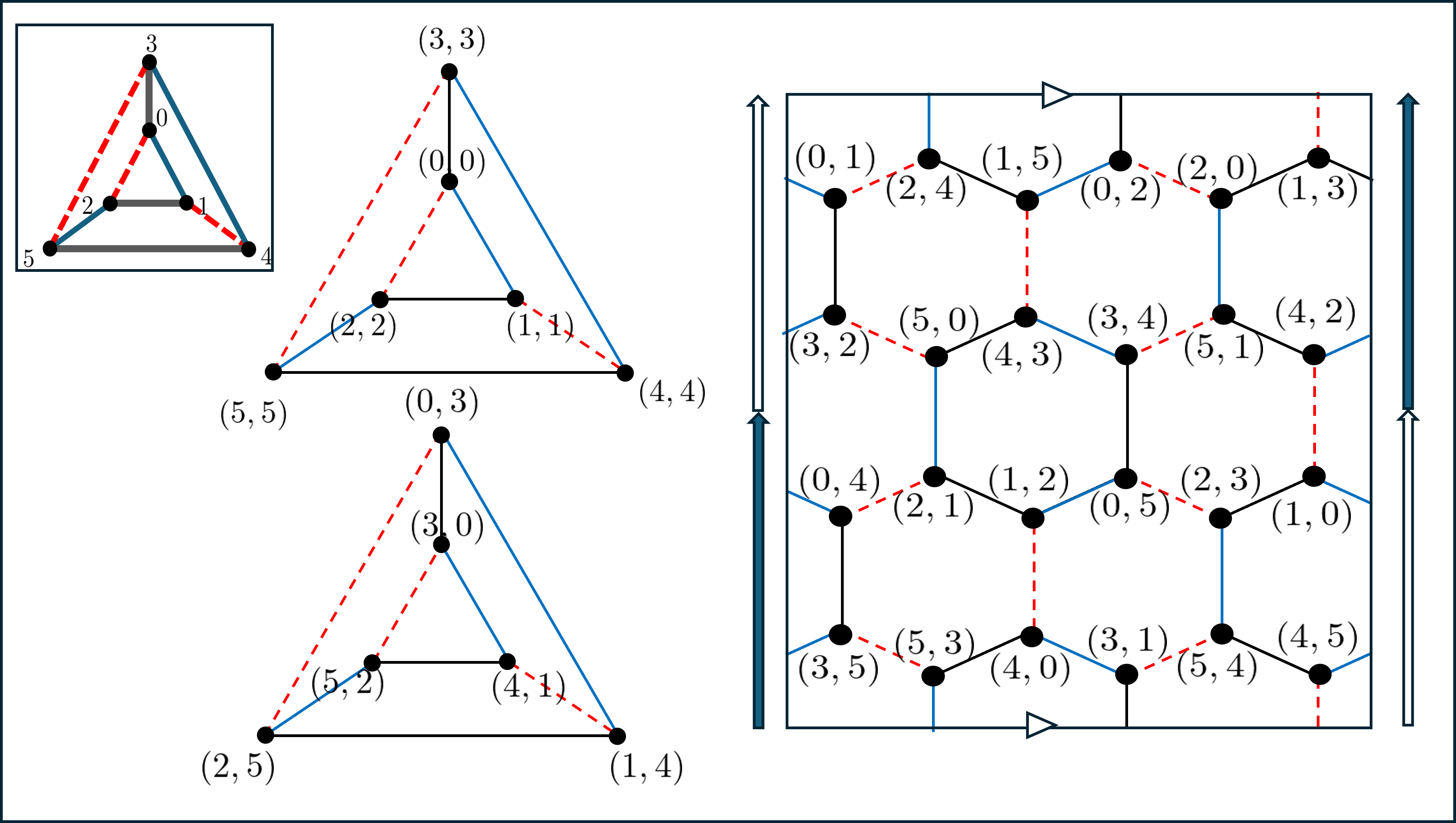}
    \subcaption{$G^{B}$'s induced by edge coloring: $\pi_c=\pi_{c,A}$}
    \label{fig:edge}
    \end{minipage}

    \vspace{3mm}
    
    \begin{minipage}{.6\linewidth}
    \centering
    \includegraphics[width=1\columnwidth]{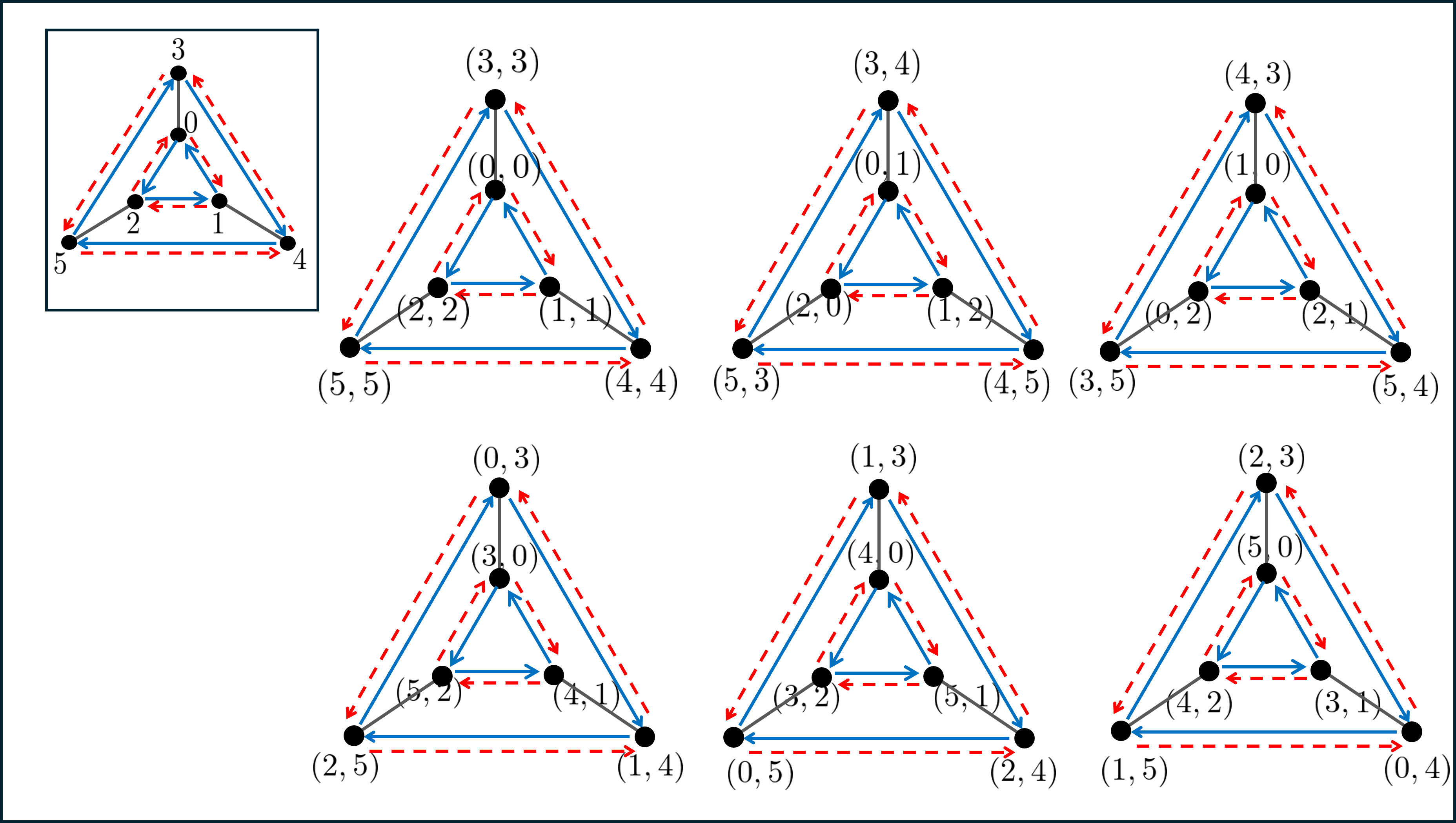}
    \subcaption{$G^{B}$'s induced by dihedral graph: $\pi_c=\pi_{c,B}$}
    \label{fig:dihedral}
    \end{minipage}

    \vspace{3mm}

    \begin{minipage}{.6\linewidth}
    \centering
    \includegraphics[width=1\columnwidth]{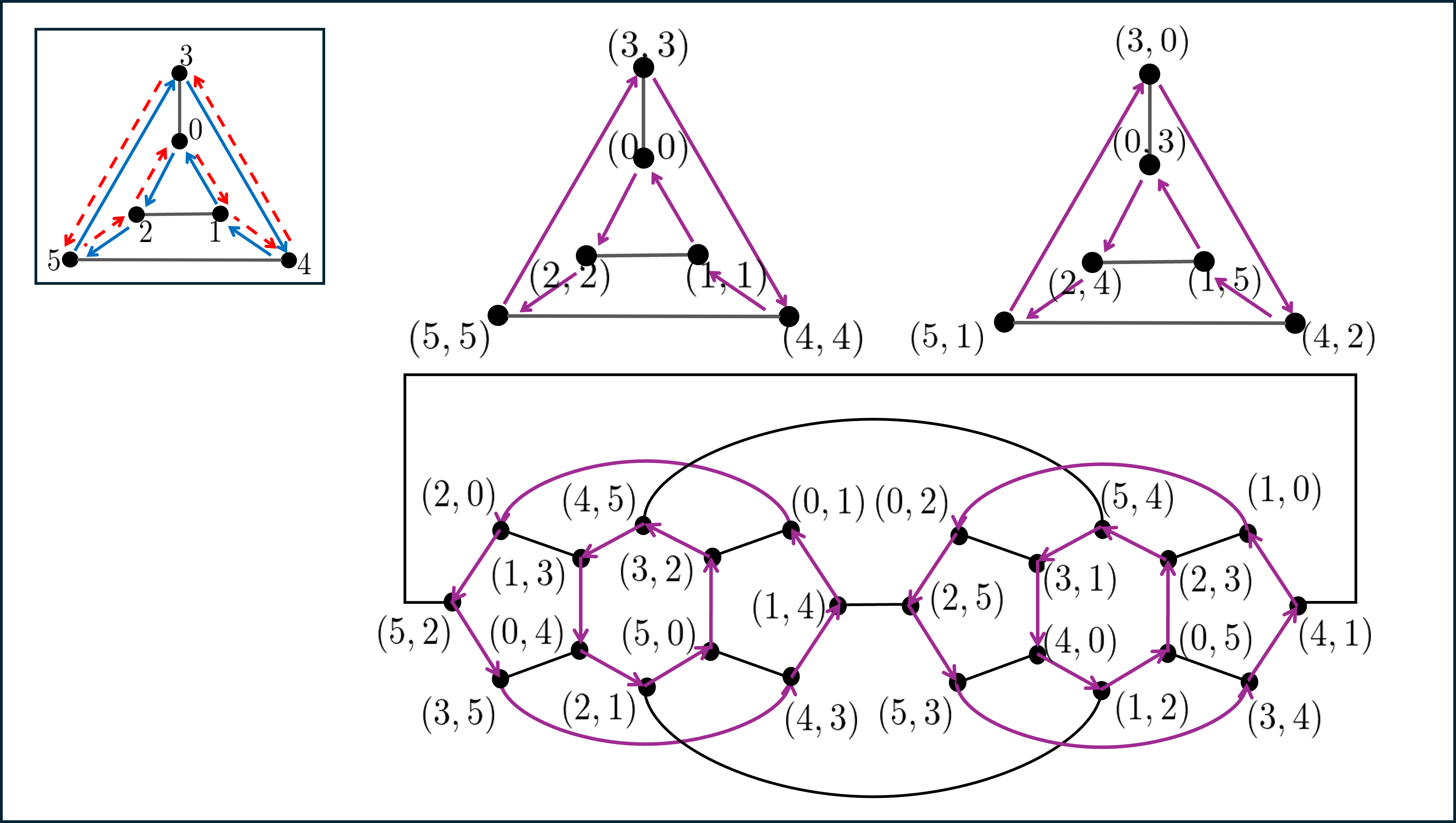}
    \subcaption{$G^{B}$'s induced by other type: $\pi_c=\pi_{c,C}$: The purple color arc 
    represents the forward directed red arc and the inverse directed blue arc. (All the graphs are symmetric digraphs.) }
    \label{fig:othertype}
    \end{minipage}
   \end{figure}

\section{Summary and discussions through numerical simulations}
We have considered an interpolating walk on finite graphs whose coherence is controlled by the parameter $q\in[0,1]$ by introducing convex combinations of discrete-time quantum walks ($q=1$) and correlated walks ($q=0$): 
\begin{equation}\label{eq:timeevolution} \mL_{p,\pi}=p\mL_\pi+q\mL_{QW}=\mL_{\pi}+q(\mL_{QW}-\mL_\pi). 
\end{equation}
The time evolution of the correlated walk $\mL_\pi$ is determined by the partition $\pi$ of the symmetric arcs of the underlying graph and the time evolution operator of the discrete-time quantum walk $\mL_{QW}$. 
We found that the spectral radius of this time evolution operator coincides with the operator norm, which is $1$. Moreover, we showed that the eigenvalues on the unit circle are semi-simple, and the eigenspace coincides with that of the  correlated walk (Theorem~\ref{prop:1}). 
By the semi-simplicity of the eigenvalues on the unit circle, the absorption state is described by the overlap of the initial state with the invariant subspace of the correlated walk (Theorem~\ref{thm:abs}). 
The absorption states of the maximal partition $\pi_{max}$ and the color partition $\pi_c$ were described more precisely, characterized by doubly stochastic random walks. 
The absorption state depends on the initial state and is determined by the graph having several connected components induced by the time evolution operator and the coloring. During the finite period, the walker can travel back and forth among these connected components by virtue of $(\mL_{QW}-\mL_\pi)$ in (\ref{eq:timeevolution}). On the other hand, during the absorption period, the walk stays in the same connected component. The ratio in each connected component is determined by the overlap with the initial state (Theorems~\ref{thm:max} and \ref{prop:abs}).  
However, the distribution converges to the uniform distribution, which is independent of the initial state and the partition (Corollaries \ref{cor:maxlimit} and \ref{cor:colorlimit}). 

Finally, we check the consistency of the statements that the interpolating walk with the partitions $\pi_{max}$, $\pi_c$ on graphs converges to the uniform distribution in the long time limit for $p>0$ if $G$ is non-bipartite (see Corollaries~\ref{cor:maxlimit} and \ref{cor:colorlimit}), and we discuss some transition behaviors to the absorption state through numerical simulations. 
To this end, we choose the underlying graph $G$ depicted in Figure~\ref{fig:arccolorling} for the cases of the maximal partition $\pi_{max}$ and the partitions induced by three kinds of arc-colorings in Figure~\ref{fig:arccolorling}, and the Grover walk whose local coin matrix is given by the Grover matrix as the full unitary walk $U$. 
Here, the Grover matrix $\mathrm{Gr}(k)$ is assigned at the vertex $x\in X$ whose degree $k$ is provided by 
\[ \mathrm{Gr}(k):=\frac{2}{k}J_k-I_k, \]
where $J_k$ is the all-ones matrix and $I_k$ is the identity matrix of size $k$. 
Let us set the arc-coloring partitions as $\pi_{c,A}$, $\pi_{c,B}$, and $\pi_{c,C}$ induced by the edge coloring type (upper figure), the dihedral group type (mid figure), and the other type (bottom figure), respectively, in Figure~\ref{fig:arccolorling}. See also Figure~\ref{fig:B}.  
The total variation distance~\cite{LP} between each $\mu_t$ and the uniform distribution $\mu_*$ is given by 
\[ d(t):=\max_{B\subset A}|\mu_t(B)-\mu_*(B)|=
\frac{1}{2}\sum_{a\in A}\left|\mu_t(a)-\frac{1}{|A|}\right|. \]
We note that $0\leq d(t)\leq 1$ for any $t\geq 0$, and if $p=1$ (full Markov chain), then $d(t)$ is non-monotonically increasing~\cite{LP} (without any oscillation). 
We plot each time course of $d(t)$ in Figure~\ref{fig:total_variation distance} with the initial state 
\[ \rho_0=\frac{1}{2}(\;|(0,1)\rangle +|(1,0)\rangle\;)(\;\langle (0,1)| +\langle(1,0)|\;) \] 
by numerical simulation. This shows the consistency of our statements that the total variation distance $d(t)$ goes to $0$ for a sufficiently large time step $t$. We take the simulation until the time $t_*:=\min\{t\;:\;d(t)\leq 0.01\}$. 
Note that if $p=0$ (full unitary), then $\mu_t$ does not converge because all eigenvalues lie on the unit circle. If $p\neq 0$, then the numerical simulations represent our result that $\mu_t$ converges to the uniform distribution, and we can see that the smaller the parameter $p$ is, the larger the final time $t_*$ is in the numerical simulation. 
In addition, there is an oscillation in the time course of $d(t)$ for $0<p<1$ due to quantum coherence. 
Although it would be better to define the mixing time by 
\[ t_{mix}:=\{ t\;:\;d_s\leq 0.01,\;\forall s>t \}, \]
we roughly regard the state as ``mixing" at the time after $t_*$ by ignoring the oscillation. 
The ``mixing" time depends on the partition $\pi$: 
for any $p\in\{0.01,0.1,0.25,0.75,1\}$, we observe that 
$t_*(\pi_{max})\leq t_*(\pi_{c,B})\leq t_*(\pi_{c,C})\leq t_*(\pi_{c,A})$.  
Therefore, one of the most interesting future problems is investigating how the transition process to the stationary state reflects the dependency on the partition and the parameter $p$ through the estimations of the ``mixing time". 

\begin{figure}[p]
\centering
\includegraphics[width=\textwidth]{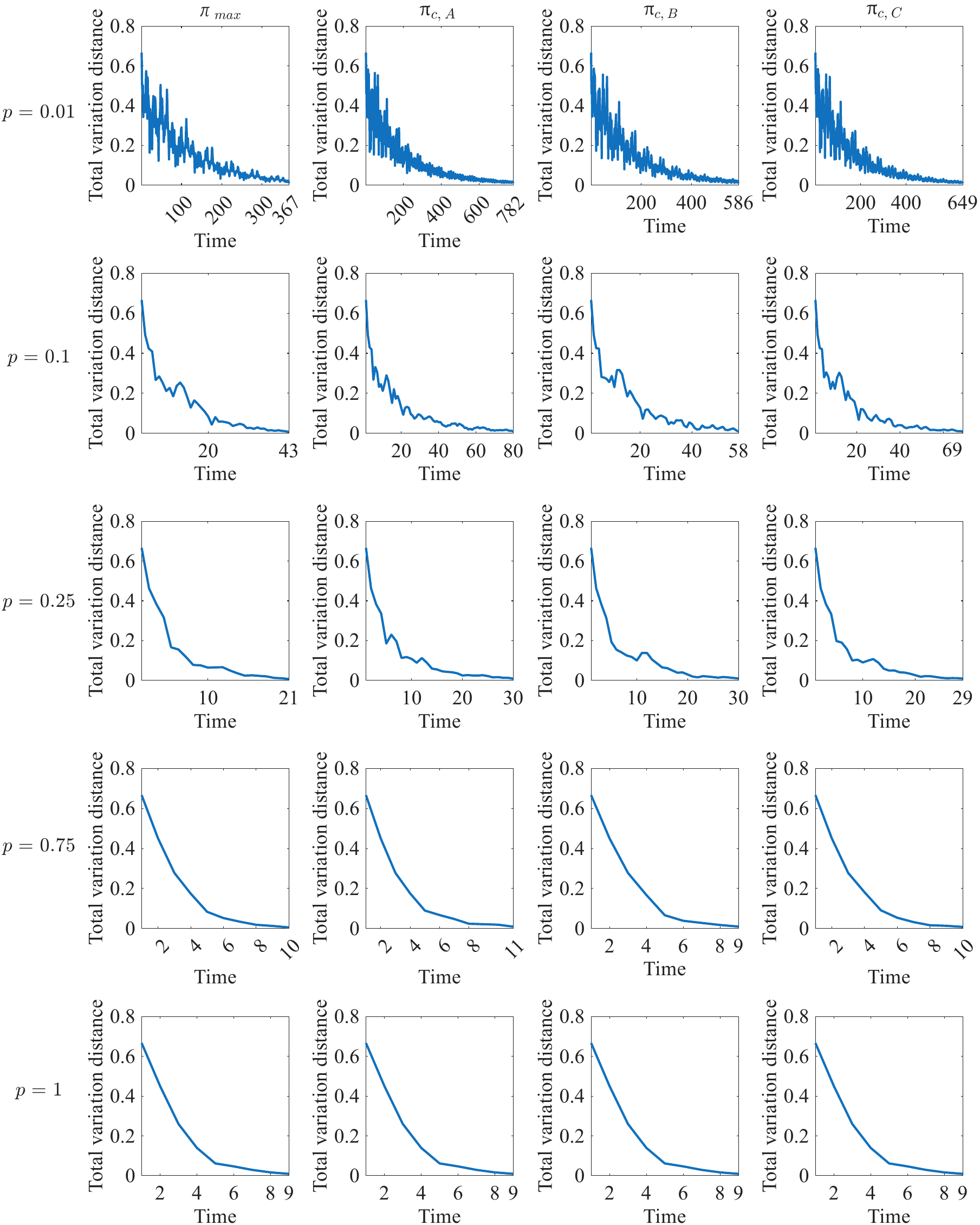}
    \caption{Time evolution of the total variation distance. The transitions until the total variation distance first drops below $0.01$ are shown for several values of $p$ and for different partitions. From left to right, the columns correspond to the four partitions $\pi_{max}$, $\pi_{c,A}$, $\pi_{c,B}$, and $\pi_{c,C}$. From top to bottom, the rows correspond to $p=0.01$, $p=0.1$, $p=0.25$, $p=0.75$, and $p=1$.}
    \label{fig:total_variation distance}
\end{figure}

\bigskip
\noindent{{\bf Acknowledgments}}
 Yu.H and E.S. acknowledge financial supports from the Grant-in-Aid of Scientific Research (C) JSPS KAKENHI Grant Nos.~23K03203, 24K06863, respectively. H.S. acknowledges financial support from the Grant-in-Aid for JSPS Fellows, JSPS KAKENHI Grant No. JP24KJ0864.

\end{document}